\documentclass[sigconf,nonacm,screen,9pt,pdfa]{acmart}

\usepackage[a-2b]{pdfx}
\usepackage{subcaption}
\usepackage{tabularx}

\usepackage{graphicx}
\usepackage{balance}
\usepackage{times}
\usepackage{tikz}
\usepackage{amsmath}
\usepackage{amsthm}
\usepackage{graphicx}
\usepackage{epstopdf}
\usepackage{latexsym}
\PassOptionsToPackage{noend}{algpseudocode}
\usepackage{algpseudocode}
\usepackage{pifont}
\usepackage{epsfig}
\usepackage{amsfonts}
\usepackage{stmaryrd}
\usepackage{url}
\usepackage{array}
\usepackage[normalem]{ulem}

\usepackage{bm}
\usepackage{cleveref}
\usepackage{listings}
\usepackage{courier}
\usepackage{enumitem}
\usepackage{appendix}

\renewcommand\footnotetextcopyrightpermission[1]{}

\usepackage[lined,boxed,vlined,ruled,linesnumbered]{algorithm2e}
\usepackage{multirow}
\usepackage{multicol}
\usepackage{listings}
\usepackage{xcolor}
\definecolor{codegreen}{rgb}{0,0.6,0}
\definecolor{codegray}{rgb}{0.5,0.5,0.5}
\definecolor{codepurple}{rgb}{0.58,0,0.82}
\definecolor{backcolour}{rgb}{0.95,0.95,0.92}
\lstdefinestyle{mystyle}{
    backgroundcolor=\color{backcolour},   
    commentstyle=\color{codegreen},
    keywordstyle=\color{magenta},
    numberstyle=\tiny\color{codegray},
    stringstyle=\color{codepurple},
    basicstyle=\ttfamily\footnotesize,
    breakatwhitespace=false,         
    breaklines=true,                 
    captionpos=b,                    
    keepspaces=true,                 
    numbers=left,                    
    numbersep=5pt,                  
    showspaces=false,                
    showstringspaces=false,
    showtabs=false,                  
    tabsize=2
}
\usepackage{xspace}
\usepackage{setspace}

\theoremstyle{definition}
\newtheorem{theorem}{Theorem}[section]

\newtheorem{definition}{Definition}[section]

\newcommand{\sys}{\textsc{SEED}\xspace}
\newcommand{\dm}{data curation\xspace}
\newcommand{\Dm}{Data curation\xspace}

\newlist{compactitem}{itemize}{3}
\setlist[compactitem]{label=\textbullet, nosep, leftmargin=0cm,itemindent=.5cm,listparindent=\parindent}

\begin{document}
\pagestyle{plain}
\title{\sys: Domain-Specific Data Curation With Large Language Models}

\author{Zui Chen$^{*}$, Lei Cao$^{+}$$^{*}$, Sam Madden$^{*}$, Tim Kraska$^{*}$, Zeyuan Shang$^{*}$, Ju Fan$^{\ddagger}$, Nan Tang$^{\dagger}$, Zihui Gu$^{\ddagger}$, Chunwei Liu$^{*}$, Michael Cafarella$^{*}$}
\affiliation{
  \institution{MIT$^{*}$, U of Arizona$^{+}$, HKUST (GZ)$^{\dagger}$, Renmin University$^{\ddagger}$}
}
\email{{chenz429,zeyuans,chunwei,kraska}@mit.edu, {lcao,madden,michjc}@csail.mit.edu, {fanj,guzh}@ruc.edu.cn, nantang@hkust-gz.edu.cn}

\begin{abstract}
\Dm tasks that prepare data for analytics are critical for turning data into actionable insights. However, due to the diverse requirements of applications in different domains, generic off-the-shelf tools are typically insufficient. As a result, data scientists often have to develop {\it domain-specific solutions} tailored to both the dataset and the task, e.g. writing domain-specific code or training machine learning models on a sufficient number of annotated examples. This process is notoriously difficult and time-consuming.
We present \sys, an {\it LLM}-{\it as}-{\it compiler} approach that automatically generates domain-specific \dm solutions via Large Language Models (LLMs). Once the user describes a task, input data, and expected output, the \sys compiler produces a hybrid pipeline that combines LLM querying with more cost-effective alternatives, such as vector-based caching, LLM-generated code, and small models trained on LLM-annotated data. \sys features an optimizer that automatically selects from the four LLM-assisted modules and forms a hybrid execution pipeline that best fits the task at hand.
To validate this new, revolutionary approach, we conducted experiments on $9$ datasets spanning over $5$ \dm tasks. In comparison to solutions that use the LLM on every data record, \sys achieves state-of-the-art or comparable few-shot performance, while significantly reducing the number of LLM calls. 
\end{abstract}

\maketitle

\settopmatter{printacmref=false}
\renewcommand\footnotetextcopyrightpermission[1]{}

\section{Introduction}
\label{sec:intro}
\begin{sloppypar}

\Dm tasks~\cite{data-tamer,bigdc} that discover, extract, transform, clean, and integrate data are critical for a wide variety of organizations. Despite significant efforts from the data management community, many sources still report that data scientists still spend over 80\% of their time on these tasks~\cite{DBLP:journals/pvldb/RezigCSSTMOTE19}. A key reason for this is that applications in different domains have diverse requirements, with no one-size-fits-all solution existing even for single \dm tasks. 
For example, for the task of data extraction, extracting monetary amounts can be effectively done by a regular expression such as that searches for a dollar sign followed by digits separated by commas and periods, i.e., $\texttt{"\$\textbackslash d[\textbackslash d|,|.]*"}$, while extracting human names requires a totally different method such as searching for capitalized words near salutations like ``Mr.'' or ``Ms.''.
Because of cases like this, generic off-the-shelf tools are rarely sufficient.
Instead, data scientists often have to develop application-specific solutions that are tailored to both the dataset and the problem domain, such as {\it domain-specific code} (like the regex above) or {\it machine learning models} trained on a large number of annotated examples to perform these types of tasks.
As a result, devising a \dm solution for a particular scenario is time-consuming, with multiple rounds of requirement generation, training data collection, model/algorithm development, and testing with both data scientists and domain experts, and rarely reusable from one deployment to the next.
This can be quite costly for enterprises -- for example, Citadel employs over 50 data management experts to deliver high-quality cleaned data to their analysts -- costing them tens of millions each year.

\subsection{Our Approach: \sys}
In this work, we propose \sys, an {\it LLM}-{\it as}-{\it compiler} approach which, allows users to describe a \dm task via a natural language specification, along with an input data. \sys automatically compiles this specification into an {\it instance-optimized} solution tailored for the data and application at hand.

The key insight is that LLMs -- with their impressive ability to generate code, perform reasoning, understand the semantics of data, and encode common knowledge -- will lead to a paradigm shift in \dm research and make it possible to automatically construct \dm solutions {\it on the fly}. Indeed, prior work has shown that LLMs can be remarkably effective at addressing specific \dm tasks~\cite{fm,evaporate,DBLP:journals/corr/abs-2306-08891}.
Unlike these prior works, which rely directly on LLMs for processing every record in a \dm task, \sys instead aims to use LLMs to generate domain-specific {\it modules} for different \dm tasks, some of which may involve direct invocation of LLMs and some of which are LLM-generated but do not use the LLM once they have been produced.

Specifically, the \sys compiler generates an execution pipeline composed of {\it code}, {\it small machine learning models}, as well as direct invocations of the LLM itself on (some) individual data records. In this execution pipeline, modules use LLMs in a variety of ways. For example, code is synthesized by the LLM to provide a domain-specific solution (e.g., a regular expression for extracting monetary amounts) and small models are trained on labels generated by the LLM. If these modules are not confident about the results on some records, \sys will forward them to the {\it LLM module}, which, although expensive, is often able to perform complex, human-like reasoning tasks on data items. For each request, the LLM module may further employ tools that retrieve relevant information from a database or other user-supplied data to assist the LLM in solving the task. Here, \sys leverages the {\it reasoning ability} of LLMs to determine on a case-by-case basis what additional information and tools will be helpful in solving the specific task.

In this way, \sys leverages LLMs' synthesis, reasoning, semantics understanding abilities as well as the encoded common knowledge to construct a domain-specific solution. 
Ideally, users do not need to manually code modules or annotate a large number of training examples. 
Moreover, unlike prior work on using LLMs for \dm tasks, \sys does not require expensive LLM invocations on every data record, which suffers from scalability, efficiency, and cost issues when handling large datasets.

\sys consists of three key components: {\bf modules}, {\bf infrastructure}, and an {\bf optimizer}.

The {\bf modules} correspond to different types of physical operators which once linked together, offer a domain specific solution to each \dm task. In our current implementation, \sys supports four types of modules: {\it CodeGen}, {\it CacheReuse}, {\it ModelGen}, and {\it LLM}. The {\it CodeGen} module uses LLM-generated code to replace the LLM in processing the data. The {\it CacheReuse} module reuses the previous exact or similar LLM query results to directly answer the new queries. The {\it ModelGen} module distills an LLM to a small machine learning model using the LLM as an annotator. The {\it LLM} module directly invokes an LLM to produce an answer. It also supports RAG-style data access with \dm tools.

\sys's {\bf infrastructure} offers a set of generic functions serving all \dm tasks. It consists of a {\it scheduler}, {\it cache storage}, and some basic optimizations to make the modules more efficient and effective. The scheduler controls the dataflow, routes data to suitable modules, and aggregates the results of all modules to produce the final answer. The cache storage, a key element of \sys, caches the input queries and LLM query results and offers interfaces for the modules to efficiently leverage the cached data to directly answer queries, train small models, validate the generated code, and augment LLMs with a RAG style solution. It features some generic optimizations such as {\it query batching} that saves LLM cost by batching a bunch of LLM queries into one and {\it tools integration} that automatically integrates the tools supplied by users into the solution. 

The {\bf optimizer} constitutes the most interesting component of \sys. Similar to database optimizers, given a specification of a \dm task, \sys’s optimizer automatically decides what modules to use, configures the chosen modules, and orders these modules into an execution pipeline. The objective is to produce a \dm plan that minimizes the execution costs with guaranteed accuracy. Because the \sys optimizer has to take both the execution time and effectiveness into consideration, its search space is much larger than a classical database optimizer where the goal is to minimize the execution time. To efficiently produce a good plan, inspired by the classic cost-based optimizers, \sys collects some statistics and then iteratively constructs the plan by adding modules one by one and reusing the optimal subplans. In addition, we leverage the unique properties of the \sys optimization problem and \dm tasks to further reduce the search space, while not missing the good plans. Moreover, the \sys optimizer dynamically re-optimizes as \sys continuously accumulates more cached data and gradually improves the performance of the cache reuse and small model. In some cases, when \sys determines the small model's performance is good enough with the labelled data it has accumulated from the LLM at a certain point in the execution, the small model might replace the LLM completely.

\vspace{-2mm}
\subsection{Contributions}
Our primary contribution is to demonstrate that LLMs enable a new, transformative approach to tackling \dm tasks. Unlike conventional \dm tools, \sys does not aim to provide pre-configured \dm models for specific data domains. Instead, based on the properties of the given dataset as well as user-specific requirements, \sys compiles an {\it instance-optimized} solution for the given dataset {\it on the fly}. During this process, users do not have to write any code or perform prompt engineering.

We have built an initial version of \sys to validate the idea. The results are encouraging: with a number of basic optimizations in the infrastructure and optimizer, which we describe below, \sys is able to produce efficient and effective solutions for multiple \dm problems using different modules, demonstrating its potential in practice. In particular, through experiments on $9$ datasets spanning various \dm tasks including data imputation, data extraction, data annotation, entity resolution, and data discovery, we show that \sys produces solutions that significantly outperform their generic counterparts, often approaching the performance of solutions trained on thousands of examples. Moreover, in comparison to solutions that use LLMs on every data record, \sys achieves state-of-the-art or comparable few-shot performance, while significantly reducing the number of required LLM calls.

Our technical contributions, mainly focused on the \sys optimizer and effectively generating modules themselves, include:
\begin{compactitem}

\item {\textbf{\sys Optimizer.}}
Inspired by database query optimizers, we have built the \sys optimizer that, given any specific \dm task, automatically produces an execution plan with minimized cost and guaranteed effectiveness. We prove that the plan is optimal under reasonable assumption. To the best of our knowledge, \sys is the first database optimizer style solution that optimizes the use of the LLMs in solving large scale data problem.

\item {\textbf {\sys Modules.}} By leveraging the cutting edge techniques in the field as well as inventing new techniques, \sys uses LLMs to effectively synthesize the \sys modules. In particular, to support scenarios with complex logic, the {\it CodeGen} module in \sys produces an ensemble of code snippets that jointly solve a task. We design an {\it evolutionary algorithm} to automatically build, filter and aggregate the ensemble of code snippets.

\item {\textbf{Infrastructure.}}
\sys's infrastructure includes a scheduler that dynamically routes the data along the execution pipeline at instance level. It also offers optimizations such as {\it query batching} and {\it tools integration} that are generally applicable to many different \dm tasks.
\end{compactitem}

\end{sloppypar}
\section{\sys Overview}
\label{sec:system}
\begin{sloppypar}

\begin{figure*}[t]
  \centering
  \includegraphics[width=1.0\linewidth]{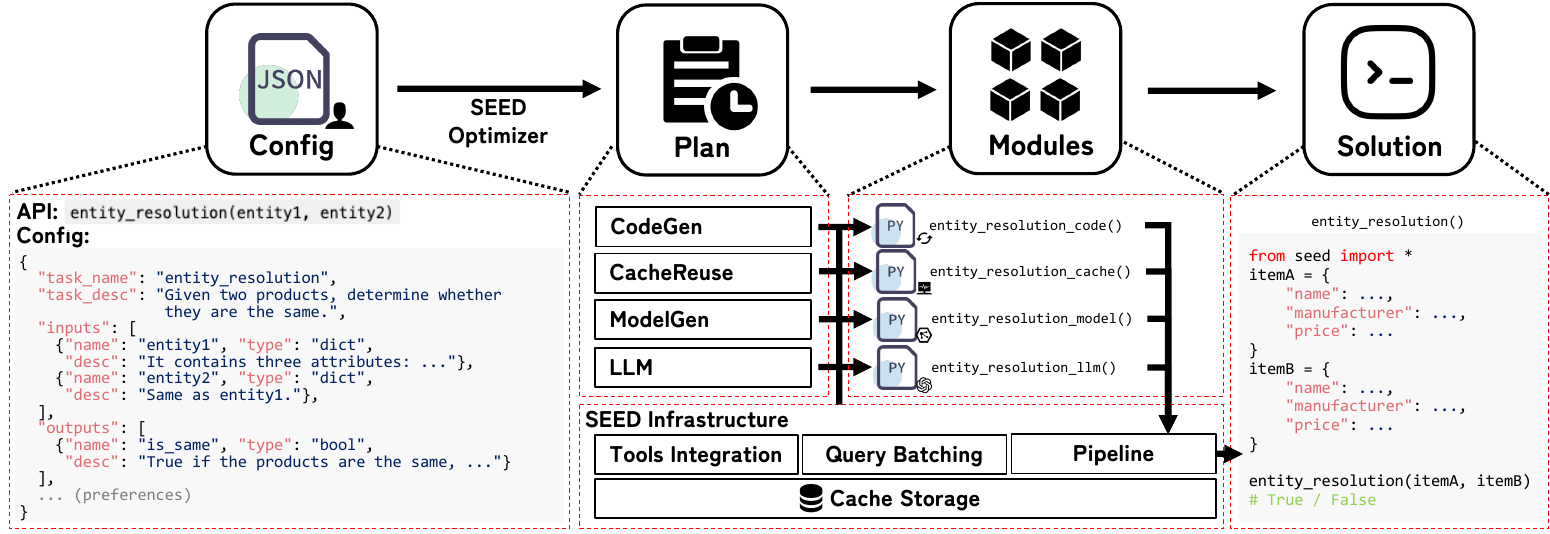}
  \vspace{-6mm}
  \caption{The Architecture of \sys, using an entity resolution task as example.}
  \label{fig:demo}
\vspace{-4mm}  
\end{figure*}

We first describe \sys from the user's perspective in Sec.~\ref{sec:system:use} and overview its internals in Sec.~\ref{sec:system:internals}.

\subsection{\sys Usage}
\label{sec:system:use}

\noindent\textbf{User Configuration.}
Users interact with \sys via a simple configuration file. Fig.~\ref{fig:demo} shows the configuration of an entity resolution task. Given a set of products with attributes {\it{``name''}}, {\it{``manufacturer''}}, and {\it{``price''}},  the user seeks to determine whether two products refer to the same item. The configuration mainly involves two parts: (1) describing the task, inputs, and outputs in natural language and (2) (optionally) providing available data and tools. The user first specifies the task name, inputs and outputs. The name can be arbitrary text. Each input also includes its name, data type, and a natural language description, -- for example,  the description may incorporate the user's domain knowledge about the attribute. The user defines the output in the same way. Optionally, the user may provide additional domain knowledge. This can be done through natural language descriptions, custom tools in the form of APIs for \sys to use, or data files containing examples, either annotated or not.

For common \dm tasks including data extraction, data discovery, entity resolution, data imputation, and data discovery, \sys offers built-in configuration templates. 
These templates automatically populate the required fields in the configuration file, which can be adjusted by users as needed. In the above entity resolution task, the template populates the input field with two entities and the output field with the following information: {\it{``name: is\_same; type: bool; description: 1 if the two titles are the same, 0 otherwise.''}}. After selecting the entity resolution template, the user only needs to optionally provide the hint and some examples.

\noindent\textbf{Compilation \& Execution.}
After the user completes the configuration, \sys compiles it into a deployable program as a functional API conforming to the user's specification, e.g., for the example in Fig.~\ref{fig:demo}, one function, {\tt {entity\_resolution(entity1, entity2) $\rightarrow$ is\_same}}, is generated. This function ingests the data in the format of data record pairs and outputs a boolean value indicating if one pair matches or not. As \sys runs, it may periodically regenerate new implementations of this function as, for example, new training data annotated by LLMs becomes available for ModelGen.

\subsection{\sys internals}
\label{sec:system:internals}

\sys compiles the solution in three steps: (1) using the \sys optimizer to produce a \dm plan that determines what modules to enable; (2) constructing the modules; (3) and linking the modules into an executable pipeline.

\noindent\textbf{(1) Producing a Data Curation Plan.}
For different tasks and datasets, the \sys optimizer produces different plans that make use of different combinations of modules in different orders. This is because different modules apply to different tasks. Blindly using all the modules on every task tends to be neither effective nor efficient.

As shown in Fig.~\ref{fig:demo}, \sys chooses and applies four modules to each task: CacheReuse, CodeGen, ModelGen, and raw LLM. 
The CodeGen module is suited to \dm tasks where most cases can be effectively addressed using rule-based methods, such as detecting errors in data that violate domain specific rules or extracting monetary amounts or human names as discussed in Sec.~\ref{sec:intro}.
The CacheReuse module performs well in the scenarios where similar or even identical instances frequently arise, while an LLM tends to generate similar or even identical results for these queries. Reusing the cached results thus avoids repeatedly invoking the LLM to answer these queries.
The ModelGen module is suitable for \dm tasks that can be modeled as predictive or generative machine learning tasks. For example, entity resolution can be viewed as a classification problem which classifies a pair of objects as matched or unmatched. Such problems can be well-addressed with a machine learning model, trained by examples generated by LLMs or supplied by users.
Finally, the LLM module, although universally applicable as advanced but expensive {\it data analysis tools}, is best suited for complex tasks that require semantic understanding ability.  Some tasks might use the LLM modules alone; for example, data discovery -- where reasoning whether a table is relevant to a user's question -- is an area where LLMs show a clear advantage over simpler methods.

Instead of leaving users to decide what modules to use, the \sys optimizer automatically produces a plan. Moreover, rather than always produce a plan with the highest possible accuracy regardless of the execution cost, \sys can trade-off accuracy to minimize execution cost by for example setting a performance gap tolerance of $10\%$, meaning that the user accepts a $10\%$ reduction in accuracy versus the maximum attainable for an increase in efficiency. Guided by this objective, the \sys optimizer produces a plan that specifies what modules to use, configures their hyperparameters, and decides the execution order.

\noindent\textbf{(2) Constructing the Modules.}
\sys uses LLMs to help construct the modules on the fly.

For plans that use the CodeGen module, \sys first translates the config file into a task description prompt using a pre-defined template. This prompt is then sent to an LLM to generate a series of code snippets, which are automatically evaluated and refined (Sec.~\ref{sec:modules-codegen}). This produces a callable method, e.g., {\tt entity\_resolution\_code()}.
Similarly, for plans with the {\it ModelGen} module, \sys uses prompt templates to request an LLM to annotate data for use in model training. Then \sys generates a callable method, e.g., {\tt entity\_resolution\_model()}, which corresponds to a small machine learning model trained on the LLM-generated training data (Sec.~\ref{sec:modules-model}).
For the {\it CacheReuse} module, given a new LLM query, \sys searches for its nearest neighbors in the cache and reuses the results if their similarity is above a threshold. 

For the {\it LLM} module, \sys generates a prompt that connects an LLM to the data access interfaces associated with the given task and composes an iterative process whereby the LLM can selectively use these interfaces to extract data. 
This also generates a callable method, e.g., {\tt entity\_resolution\_llm()}, which implements a RAG style \dm solution (Sec.~\ref{sec:infra-tools}).
 
Note although not the major focus of this paper, we have spent significant effort in \sys to carefully design templates for dynamic prompt composition that perform well in general on these types of \dm tasks, thus eliminating the need for prompt engineering by users. For the prompts templates please see Appendix~\ref{sec:appendix:prompts} of the extended version~\cite{seed}.

\noindent\textbf{(3) Linking and Execute the Modules.}
Eventually, \sys links all these internal method calls above to a final method, e.g., {\tt entity\_resolution(entity1, entity2)} {\tt$\rightarrow$ is\_same}. The generated method sequentially executes the modules in an order specified by the \dm plan over the data records. Each module determines for each record whether it should directly return the result or pass it to the next module. For example, for CodeGen, generated code snippets are explicitly instructed to refrain from answering when faced with uncertain cases, in which case execution will continue with the next module, while CacheReuse and ModelGen refrain from answering when they produce predictions with low confidence (Sec.~\ref{sec:modules-model}). The LLM module, which typically is the last step in the execution pipeline, always answers, using the available data access tools to retrieve relevant information from data to improve response accuracy.

\end{sloppypar}
\section{\sys Optimizer}
\label{sec:optim}
\begin{sloppypar}
Given a \dm task, the \sys optimizer automatically produces an optimized execution plan $\mathcal{P}$ which selects a set of modules $M_i$ to activate, configures their hyperparameters $\theta_i$, and decides their execution order in the pipeline. The objective is to minimize the overall execution cost $\mathcal{C}$ while yielding effectiveness $\mathcal{A} \in [0,1]$ that is close to that of the most effective plan. 

\begin{definition}[\textbf{\sys Optimization}]
\label{thm:seed-optim-def}
Given an effectiveness gap $G$ that the users can tolerate, the \sys optimizer targets finding an execution plan $\mathcal{P} = [\mathbb{M} (\mathbb{\theta}), O(\mathbb{M})]$, that among all possible plans, has the lowest cost $\mathcal{C}$ and has an effectiveness $\mathcal{A}(\mathcal{P})$ > $\mathcal{A}(\mathcal{P}^*) - G$. Here $\mathcal{P}^*$ denotes the plan with highest effectiveness $\mathcal{A}$; $\mathbb{M} (\mathbb{\theta})$ denotes the set of selected modules $M_i$ and their corresponding hyperparameters $\theta_i$; and $O(\mathbb{M})$ denotes the order of these modules.
\end{definition}

For example, if the user sets the effectiveness gap to $10\%$, the \sys optimizer will produce a plan with minimal execution costs subject to the constraint that it is at most 10\% worse than a plan with highest possible effectiveness. Tab.~\ref{tab:hyperparameters} shows the hyperparameters $\theta_i$ of each module we currently support.

The \sys optimization problem bears similarity to database query optimization at a high level, in that both choose the physical implementation of operators and decides their ordering. However, \sys optimization is different in several respects. First,  database query optimization targets minimizing the execution cost of relational queries, while \sys optimization has to take both execution cost and effectiveness into consideration, thus leading to a larger search space. This is because in relational databases, all valid query plans produce the same results, while in \sys, given a \dm task such as entity resolution, different plans that involve AI models might produce different results with different accuracies. Second, in \sys, the optimal \sys plan evolves over time, because the accuracy of modules like CacheReuse and ModelGen tend to keep improving as more results are produced and cached during execution time. Therefore, \sys requires a dynamic re-optimization approach to adapt the plan.

\begin{table}[htbp]
  \vspace{-2mm}
  \caption{\sys Module Hyperparameters}
  \vspace{-3mm}
  \label{tab:hyperparameters}
  \begin{tabularx}{0.7\linewidth}{ll}
    \toprule
    \multirow{2}{*}{ CacheReuse   }     & Activate                      \\
                                        & Distance Threshold            \\
    \midrule
    \multirow{3}{*}{ CodeGen }          & Activate                      \\
                                        & Num Branches                  \\
                                        & Num Preserved Branches        \\
                                        & Num Iterations                \\
    \midrule
    \multirow{2}{*}{ ModelGen   }       & Activate                      \\
                                        & Confidence Threshold          \\
    \midrule
    \multirow{2}{*}{ LLM     }          & Activate                      \\
                                        & Examples Sample Mode          \\
    \bottomrule
  \end{tabularx}
  \vspace{-4mm}
\end{table}

\subsection{Generic \sys Optimizer}
\label{sec:optim-general}

To generate an efficient and effective plan $\mathcal{P}$, inspired by the classic Selinger optimizer, the \sys optimizer uses dynamic programming to reduce the search space. That is, \sys adds modules one by one, keeping track of the selected subplans in a memoization table to avoid recomputing subplans. 
However, unlike the Selinger optimizer, which only has to store one lowest-cost subplan for each subquery, \sys has to decide what subplans to store based on both effectiveness and cost. Clearly, \sys cannot keep one single plan with either lowest cost and highest effectiveness. However, it cannot afford to store all plans  due to the memory and the search costs. To solve this problem, we propose a skyline-based method which stores a subplan  {\it if and only if} it potentially could be a part of the final full plan. More specifically, \sys will never keep a subplan unless it is a {\it skyline subplan}, where the concept of {\it skyline subplan} is defined in Def.~\ref{def.skyline}.

\begin{definition}[\textbf{Skyline Plan}]
\label{def.skyline}
A subplan $\mathcal{P}_i$ dominates the other subplan $\mathcal{P}_j$, if $\mathcal{C}(\mathcal{P}_i) < \mathcal{C}(\mathcal{P}_j)$ and $\mathcal{A}(P_i) > \mathcal{A}(\mathcal{P}_j)$. A subplan $\mathcal{P}_i$ is a skyline subplan if $\mathcal{P}_i$ is \textbf{NOT} dominated by any subplan $\mathcal{P}_j$.
\end{definition}

Intuitively, a subplan $P_i$ is a skyline plan if there does not exist any plan that is better than it on both effectiveness and cost. If a subplan is not a skyline plan, it will never be a part of the final plan. This is because replacing it with a skyline plan that dominates it will end up with an equal or better plan. 

\begin{algorithm}[t!]
    \setstretch{0.60}
    \caption{Generic SEED Optimizer}\label{alg:seed-optim}
	\LinesNumbered
    \smaller
	\KwIn{Task, Gap $G$}
	\KwOut{Configuration $\Theta$}
        $f[0][0] \leftarrow \mathcal{P}_{\varnothing}$; \\
        \For {$i \leftarrow 1$ to $+\infty$} {
            $f[i] \leftarrow f[i-1]$; \\
            \For{$\mathcal{P}^{i-1} \in f[i-1]$} {
                \For {$M \in modules$ and $M$ is not used in $\mathcal{P}^{i-1}$}{
                    \For{$\theta_M \in \texttt{Grid(M)}$}{
                        $\mathcal{P}^{i} \leftarrow \texttt{append}(\mathcal{P}^{i-1},\theta_M)$; \\
                        $f[i][\mathcal {A}(\mathcal{P}^{i})] \leftarrow \arg\min_{\mathcal{P} \in \lbrace f[i][\mathcal {A}(\mathcal{P}^{i})], \mathcal{P}^{i} \rbrace} \mathcal {C}(\mathcal{P})$; \\
                    }
                }
            }
            \For {$\mathcal{P} \in f[i]$} {
                $f[i][\mathcal {A}(\mathcal{P})] \leftarrow \arg\min_{\mathcal{P} \in \lbrace f[i][a] \vert a \ge \mathcal {A}(\mathcal{P}) \rbrace} \mathcal {C}(\mathcal{P})$; \\
            }
            \If {$f[i] = f[i-1]$} {
                break; \\
            }
        }
        $\mathcal{P}^\ast \leftarrow \arg\max_{\mathcal{P} \in f[i]} {\mathcal{A}}(\mathcal{P})$; \\
        return $\arg\min_{\mathcal{P} \in f[i]} {\mathcal{C}}(\mathcal{P})$ s.t.: ${\mathcal{A}}(\mathcal{P}^\ast) - {\mathcal{A}}(\mathcal{P}') \le G$; \\
\end{algorithm}

Next, we describe how the \sys optimizer works in more detail. As shown in Alg.~\ref{alg:seed-optim}, the \sys optimizer iteratively attempts to append a new module to the existing subplans (Alg.~\ref{alg:seed-optim}, L4$\sim$5), and then discards the dominated subplans (Alg.~\ref{alg:seed-optim}, L9$\sim$10). This process continues until no better plans can be found. Unlike exhaustive search, this approach avoids recomputing previous subplans and prunes invalid and dominated subplans as well as their descendants. Finally, \sys optimizer stores all the skyline plans after considering all the modules, and then by default returns the best plan $\mathcal{P}$, that optimizes $\min_\mathcal{P} {\mathcal{C}(\mathcal{P})}$ subject to: $\mathcal{A}(\mathcal{P}^\ast) - \mathcal{A}(\mathcal{P}) \le G$ (Alg.~\ref{alg:seed-optim}, L14). 

The \sys optimizer does not use a greedy algorithm that optimizes the configuration of each module independently. Instead, when attempting to add a new module $M_1$ into an existing subplan $\mathcal{P} = [\langle M_2(\theta_2), M_3(\theta_3) \rangle]$, the \sys optimizer will reuse this subplan including the parameter configurations of $M_2$ and $M_3$ and their relative order, but decide the configuration of $M_1$ based on the collective performance $\mathcal{A}(\mathcal{P}')$ of the pipeline composed of the modules in the newly formed plan $\mathcal{P}'$. For example, if $\mathcal{P}' = [\langle M_2(\theta_2), M_3 (\theta_3), M_1(\theta_1) \rangle]$, the effectiveness $\mathcal{A}(\mathcal{P}')$ is measured by executing the pipeline that applies $M_2$, $M_3$, and then $M_1$ on a validation dataset.

As shown in Fig.~\ref{fig:optim}, the \sys optimizer may start by choosing the LLM module, deciding its best hyperparameter configuration, and storing it as the current optimal subplan. Then the \sys optimizer attempts to append different modules to the single LLM module subplan. It finds that appending CodeGen or ModelGen in front of the LLM module is able to advance the Pareto frontier. Thus, both $\mathit{[\langle CodeGen, LLM\rangle]}$ and $\mathit{[\langle ModelGen, LLM\rangle]}$ are maintained as valid subplans. Then the \sys optimizer may continue to try plugging more modules. Finally, when the plans stop improving, the \sys optimizer returns the most efficient plan from all the valid plans (i.e., plans where $\mathcal{A}(\mathcal{P}^\ast) - \mathcal{A}(\mathcal{P}) \le G$) on the Pareto frontier.

Next, we show that as long as the descendants of a dominated subplan are also dominated during search, pruning the dominated subplans and their descendants does not break the optimal substructure property of dynamic programming and thus preserves the optimality of the \sys optimizer. 

\begin{theorem}
\label{thm:seed-optim-exact}
As long as the descendants of a dominated subplan are also dominated, the \sys optimizer is guaranteed to find the optimal solution $\mathcal{P}^\ast = \max_{\mathcal{P}}\ \mathcal{A}(\mathcal{P})$ as well as the most efficient solution $\min_{\mathcal{P}}\ \mathcal{C}(\mathcal{P})$ subject to: $\mathcal{A}(\mathcal{P}^\ast) - \mathcal{A}(\mathcal{P}) \le G$.
\end{theorem}

\begin{proof}
We will prove by contradiction. Assume the \sys optimizer eliminated $\mathcal{P}^\ast$. Then there must exist a configuration $\mathcal{P}$ which is an ancestor of $\mathcal{P}^*$ (i.e., $\mathcal{P}^\ast = \mathcal{P} \oplus \mathcal{P}_\Delta$ for some $\mathcal{P}_\Delta$ representing the appended modules during the transition from $\mathcal{P}$ to $\mathcal{P}^\ast$), and $\mathcal{P}$ was eliminated because it is dominated by some other plan $\mathcal{P}'$. Since the descendants of a dominated subplan are also dominated, $\mathcal{P}' \oplus \mathcal{P}_\Delta$ dominates $\mathcal{P}^\ast$. This contradicts the optimality of $\mathcal{P}^\ast$. Therefore, by proof of contradiction, the optimal plan $\mathcal{P}^\ast$ will not be pruned by the \sys optimizer. Since no optimal plan is removed, the \sys optimizer preserves optimality. 

The same proof applies to $\min_{\mathcal{P}}\ \mathcal{C}(\mathcal{P})$, as the filter $\mathcal{A}(\Theta^\ast) - \mathcal{A}(\Theta) \le G$ only takes place at the end of the search.
\end{proof}

\begin{figure}[t]
  \centering
  \vspace{-2mm}
  \includegraphics[width=1.0\linewidth]{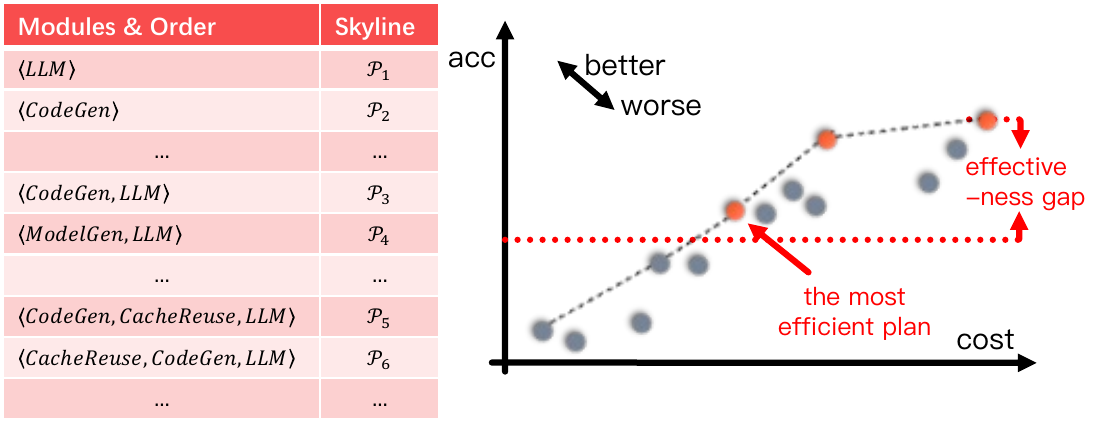}
  \vspace{-6mm}
  \caption{Illustration of the \sys Optimizer.}
  \vspace{-4mm}
  \label{fig:optim}
\end{figure}

Note the assumption that the descendants of a dominated subplan are also dominated typically holds. For example, if the subplan $\mathcal{P}$ = $[\langle M_2(\theta_2), M_3(\theta_3) \rangle]$ composed of modules $M_2$ and $M_3$ is dominated by another subplan $\mathcal{P}'$ = $[\langle M_2(\theta'_2), M_3(\theta'_3) \rangle]$ which is composed of modules $M_2$ and $M_3$ as well, but uses different parameter configurations. 
Appending a new module $M_1$ will migrate $\mathcal{P}$ to $[\langle M_2(\theta_2), M_3(\theta_3), M_1(\theta_1)\rangle]$ and $\mathcal{P}'$ to $[\langle M_2(\theta'_2), M_3(\theta'_3), M_1(\theta'_1) \rangle]$. It is unlikely that $\mathcal{P}$ will suddenly become better than $\mathcal{P}'$, especially because \sys optimizer will configure module $M_1$ differently in $\mathcal{P}$ and $\mathcal{P}'$ to maximize their performance boost. 
Note the Selinger optimizer makes the similar assumption. However, unlike the Selinger optimizer which only keeps one optimal subplan for each subquery, the \sys optimizer potentially could store multiple skyline plans for the same subset of modules. This further reduces the chance that the descendants of a dominated subplan will become the optimal plan at the end.  

\subsection{Specialized \sys Optimizer}
\label{sec:optim-specialized}
Although it reduces the search space by reusing subplans, the generic \sys optimizer is still very time-consuming. Therefore, we propose to further optimize its efficiency by leveraging the unique properties of the \sys execution pipeline and \dm tasks. 
In particular, we design a specialized optimizer based on the following optimizations: 1) Prioritizing the ordering of the modules; 2) Incorporating domain-specific constraints to restrict the search space; 3) Accelerating the hyperparameter search for each module; and 4) Cached validation.

\noindent\textbf{Optimization \#1: Prioritizing the Module Ordering.} When adding a module into an existing subplan and searching for its optimal hyperparameters, the generic \sys optimizer has to enumerate all possible orders, resulting in a large search space. To address this issue, \sys introduces a strategy to prioritize the module ordering before starting to search for the hyperparameters of the new module. It is based on our observation that due to the fallback mechanism of the \sys pipeline, the ordering of the modules often does not matter much to the accuracy. As an example, suppose a pipeline is composed of the ModelGen module and the LLM module, where the LLM module usually has a higher accuracy than the ModelGen module. As long as ModelGen could appropriately fallback to LLM, running ModelGen ahead of LLM would not necessarily lead to an accuracy lower than running LLM first. Moreover, with a limited number of validation examples, in fact it is rather difficult to accurately measure the effectiveness of each possible subplan. Therefore, it seems unpromising to enumerate a larger search space, but only reap a marginal accuracy gain which could possibly be canceled out by the estimation error. 

On the other hand, the ordering indeed matters to the execution cost. Clearly, executing ModelGen first and only fallbacking to LLM on a small number of data records will be much more cost effective than routing all data records to LLM first. 

Therefore, in this work, we propose a strategy that efficiently produces an order to minimize the execution cost of the pipeline.

We start with defining a way to quickly estimate the execution cost of a \sys plan. 

\begin{definition}[The Execution Cost of \sys Plan]
\label{def:seed-optim-def}
Each module $M_i$ has a fixed hyperparameter configuration $\theta_i$ which determines its execution cost $\mathcal{C}_i$ and the fallback probability $p_i \in [0,1]$ that $M_i$ will route a tuple to the next module in the pipeline. An execution plan $\mathcal{P}$ consists of $n$ modules ordered in the sequence $\langle M_i \rangle_{i=1}^n$. The total estimated execution cost of the plan $\mathcal{P}$ can be calculated as:
\vspace{-2mm}
\begin{equation}
\label{eq.cost}
\mathcal{C}(\mathcal{P}) = \sum_{i=1}^n \left[\left(\prod_{j=1}^{i-1} p_j\right) \mathcal{C}_i\right]
\end{equation}

\end{definition}

In Def.~\ref{def:seed-optim-def}, for each module $M_i$, its execution cost $\mathcal{C}_i$ and fallback ratio $p_i$ are estimated independently. Thus, \sys only has to compute these statistics once and repeatedly uses them in the optimization process. The estimation of $\mathcal{C}_i$ and $p_i$ can utilize both labeled validation data and unlabelled data collected during execution, as ground truth is not required for cost and fallback ratio estimation.

\vspace{-4mm} 
\begin{equation}
\label{eq.score}
\mathit{priority(M_i)} = \frac{1-p_i}{\mathcal{C}_i}
\end{equation}

Once these statistics are obtained, the module ordering algorithm in \sys uses Eq.~\ref{eq.score} to calculate the priority for each module and then sorts the modules in descending order of their priorities.

Next, we prove that this simple sorting algorithm indeed minimizes the total execution cost of a plan.

\begin{theorem}
\label{thm:seed-optim-order}
For a set of modules $\mathbb M(\theta)$, sorting the modules $M_i \in \mathbb M$ in descending order of $\mathit{priority(M_i)}$ will result in a plan $\mathcal P = [\mathbb M(\theta), O_{\mathit{priority}}(\mathbb M)]$ that minimizes $\mathcal{C}(\mathcal{P})$.
\end{theorem}
\begin{proof}
Given two adjacent modules $M_1$ and $M_2$, swapping $M_1$ and $M_2$ would not impact the execution cost of the modules before or after them in the pipeline. Thus, consider two adjacent modules $M_1$: $(p_1,C_1)$ and $M_2$: $(p_2,C_2)$, denoting the execution cost of the modules before $M_1$ and $M_2$ as $x$, then module $M_1$ should be in front of module $M_2$ in the optimal ordering, if and only if:
$$
\begin{aligned}
p_1 \times (p_2 \times x+{\mathcal{C}}_2)+{\mathcal{C}}_1 &\le p_2\times (p_1 \times x+{\mathcal{C}}_1)+{\mathcal{C}}_2 \\
\Leftrightarrow p_1 \times p_2 \times x + p_1 \times {\mathcal{C}}_2+{\mathcal{C}}_1 &\le p_1 \times p_2 \times x + p_2 \times {\mathcal{C}}_1+{\mathcal{C}}_2 \\
\Leftrightarrow p_1 \times {\mathcal{C}}_2+{\mathcal{C}}_1 &\le p_2 \times {\mathcal{C}}_1+{\mathcal{C}}_2 \\
\Leftrightarrow \frac{p_1}{{\mathcal{C}}_1}+\frac{1}{{\mathcal{C}}_2} &\le \frac{p_2}{{\mathcal{C}}_2}+\frac{1}{{\mathcal{C}}_1} \\
\Leftrightarrow \frac{1-p_2}{{\mathcal{C}}_2} &\le \frac{1-p_1}{{\mathcal{C}}_1} \\
\Leftrightarrow priority(M_2) &\le priority(M_1) \\
\end{aligned}
$$

Then by considering bubble sort, which only swaps the adjacent elements at each step, we can conclude that when all modules are sorted by $(1-p_i)/\mathcal{C}_i$ in descending order, $\mathcal{C}(\mathcal{P})$ is minimized.
\end{proof}

\begin{algorithm}[t!]
    \setstretch{0.60}
    \caption{Specialized SEED Optimizer}\label{alg:seed-optim-2}
	\LinesNumbered
    \smaller
	\KwIn{Task, Gap $G$, Init $I \in \lbrace True, False \rbrace$, Beam size $B$}
	\KwOut{Configuration $\mathcal{P}$}
        $f[0][0] \leftarrow \mathcal{P}_{\varnothing}$; \\
        \For {$i \in [LLM, CodeGen, ModelGen, CacheReuse]$} {
            \If {$f[i]$ is cached} {
                $f[i] \leftarrow \texttt{load\_cache}()$; continue; \\
            }
            $f[i] \leftarrow f[i^{prev}]$; \\
            \For{$\mathcal{P}^{prev} \in f[i^{prev}]$} {
                \For{$\theta_i \in \texttt{Default/Grid/Ternary}(M_i)$}{
                    $\mathcal{P}^{i} \leftarrow \texttt{sort\_by\_priority}(\texttt{append}(\mathcal{P}^{prev},\theta_M))$; \\
                    $f[i][\mathcal {A}(\mathcal{P}^{i})] \leftarrow \arg\min_{\mathcal{P} \in \lbrace f[i][\mathcal {A}(\mathcal{P}^{i})], \mathcal{P}^{i} \rbrace} \mathcal {C}(\mathcal{P})$; \\
                }
            }
            $\mathcal{P}^\ast \leftarrow \arg\max_{\mathcal{P} \in f[i]} {\mathcal{A}}(\mathcal{P})$; \\
            \For {$\mathcal{P} \in f[i]$ s.t.: ${\mathcal{A}}(\mathcal{P}^\ast) - {\mathcal{A}}(\mathcal{P}') > G$} {
                $f[i].\texttt{delete}(\mathcal{P})$; \\
            }
            \For {$\mathcal{P} \in f[i]$} {
                $f[i][\mathcal {A}(\mathcal{P})] \leftarrow \arg\min_{\mathcal{P} \in \lbrace f[i][a] \vert a \ge \mathcal {A}(\mathcal{P}) \rbrace} \mathcal {C}(\mathcal{P})$; \\
            }
            \If {$\vert f[i] \vert > B$} {
                $f[i] \leftarrow TopB_{\mathcal{A}}(f[i])$; \\
            } 
            \If {$i = LLM$ or $i = CodeGen$} {
                $\texttt{save\_cache}(f[i])$; \\
            }
        }
        return $\arg\min_{\mathcal{P} \in f[i]} {\mathcal{C}}(\mathcal{P})$; \\
\end{algorithm}

This insight allows the \sys optimizer to focus on the configuration of the modules rather than their ordering, e.g. determining if it should activate a module or not and selecting its hyperparameters. That is, it adds the modules into the plan one at a time and estimates the cost and effectiveness of this subplan based on the execution order computed with Theorem~\ref{thm:seed-optim-order}, and decides if this subplan should be kept based on the skyline rule described in Sec.~\ref{sec:optim-general}. This is equivalent to directly inserting a new module in the appropriate position based on its priority, without enumerating all possible orders.
Note that during optimization the order that considers the modules does not have to comply to their actual execution order determined above. In practice, by optimizing modules in descending order of their estimated effectiveness, \sys optimizer is able to find strong subplans at the early iterations. This yields a tighter Pareto frontier earlier in the process (Alg.~\ref{alg:seed-optim-2}, L11$\sim$12). In this way, the \sys optimizer is able to prune the suboptimal plans more effectively as the optimization proceeds.

\noindent\textbf{Optimization \#2: Task-specific Constraints.}
Secondly, to further reduce the search space, the \sys optimizer implements a set of predefined rules tailored for \dm tasks as well as common sub-types of \dm tasks, like data imputation or data extraction. For example, the hyperaparameters for code generation are typically agnostic to the specific \dm problem, suggesting that \sys could reduce the search space by pre-defining a set of default configurations (Alg.~\ref{alg:seed-optim-2}, L7) for the CodeGen module and pick one from them during optimization. In practice, \sys restricts CodeGen configuration to three options: 1) not activating CodeGen, 2) using a single code snippet, or 3) using an ensemble with a fixed number of branches. As another example, for data discovery tasks where LLMs show clear advantage over other Modules in effectiveness, \sys disables the CodeGen, CacheReuse and ModelGen modules by default, leaving only the LLM module in play.

\noindent\textbf{Optimization \#3: Improved Hyperparameter Search.}
Thirdly, to configure the hyperparameters within a specific module, various approximated search approaches can be incorporated to speed up the search. For example, for continuous parameters, although a large step-size grid search is often adequate in most cases, when a more fine-grained search is desirable, \sys can employ ternary search~\cite{ternary} to effectively reduce the search space (see Alg.~\ref{alg:seed-optim-2}, L7). Another case is when the number of configurations on the Pareto frontier become excessive, a beam search-like approach can be employed to retain only the configurations with the best performance (Alg.~\ref{alg:seed-optim-2}, L15$\sim$16). This allows \sys to trade off the optimality of the optimizer for a further restriction on the number of configurations to search.

\noindent\textbf{Optimization \#4: Cached Validation.} As mentioned above, given a plan, \sys evaluates its effectiveness $\mathcal{A}$ by executing the plan on a validation dataset. However, evaluating all plans in this way can be extremely costly. This is particularly true for the LLM module, as repeated issuing LLM calls would be prohibitively expensive. To overcome this challenge, we develop a caching mechanism. It caches all input-output pairs for each module $M_i(\theta_i)$. Consequently, if the same $M_i(\theta_i)$ appears in multiple plans, the cached results will be reused. This approach significantly reduces the plan evaluation cost of the \sys optimizer.

\subsection{Dynamic Optimization}
\label{sec:optim-dynamic}

As mentioned above, instead of generating one single static plan, the \sys optimizer dynamically re-optimizes in response to the changes on the performance of the modules. 
After initial deployment, \sys will accumulate more data in cache storage. This will benefit the CacheReuse module as well as the ModelGen module through retraining with new data. Specifically, in the beginning when no labelled data is available, the initial plan will only consider the LLM and CodeGen modules. However, after \sys accumulates a certain amount of new data and improves the performance of the CacheReuse and ModelGen modules remarkably, \sys will trigger a re-optimization to regenerate the plan, thus continuously improving the quality of the execution plan.

Rather than re-running the optimization from scratch, the \sys optimizer leverages the results from the previous optimization rounds to reduce search time. By caching and reusing the prior configurations where possible (Alg.~\ref{alg:seed-optim}), it updates the plans more efficiently.

Firstly, because the LLM module and the CodeGen module typically do not benefit much from the accumulated data, \sys only needs to compile these two modules once. Therefore, the \sys optimizer is able to directly reuse their existing configurations during subsequent optimization (Alg.~\ref{alg:seed-optim-2}, L17$\sim$18). 

Secondly, during the execution time, assuming the data distribution only changes gradually over short periods, any required modifications to the optimal plan are also likely to be minor. Rather than performing a full grid or ternary search across the entire hyperparameter space of CacheReuse and ModelGen, the optimizer restricts the search space of these continuous hyperparameters to a narrow range centered around the values from the previous optimal plan. This reduced search space exploits the incremental nature of distribution shifts, expediting the optimization process.

\noindent\textbf{Putting It All Together.}
After incorporating all the optimizations in Sec.~\ref{sec:optim-specialized} and Sec.~\ref{sec:optim-dynamic},  the final \sys optimizer (Alg.~\ref{alg:seed-optim-2}) starts by jointly optimizing the LLM module and CodeGen module and keeps a Pareto frontier. This Pareto frontier is computed once during the initial optimization process and reused for further re-optimizations (Algorithm~\ref{alg:seed-optim-2}, L17$\sim$18). Then during the execution, the CacheReuse and ModelGen modules are periodically re-evaluated, incorporating improvements such as approximate search techniques and heuristics derived from domain knowledge and the previous optimization rounds. This dynamic process allows the \sys optimizer to efficiently generate plans with guaranteed performance.

\end{sloppypar}

\section{\sys Modules}
\label{sec:modules}
\begin{sloppypar}
\sys supports three types of modules that use LLMs in various ways: CacheReuse, CodeGen, and ModelGen. Although each module draws inspiration from the existing techniques such as reusing the cached LLM outputs~\cite{gpt-cache}, \sys is the first system that leverages these techniques to produce an execution pipeline instance optimized for a \dm task. Moreover, we propose new techniques to make the {\it CodeGen} and {\it ModelGen} modules better serve \dm tasks, which are introduced in Sec.~\ref{sec:modules-codegen} and Sec.~\ref{sec:modules-model}.

\subsection{CodeGen Module}
\label{sec:modules-codegen}

Although using LLMs to generate code has attracted a lot of attention~\cite{metagpt,coderl,codet,recode,agentcoder}, LLMs still struggle to generate large-scale applications or programs requiring complex logic. However, some \dm tasks may require implementing complex logic corresponding to multiple combined rules. For instance, data extraction that extracts multiple types of targeted information from unstructured data (e.g., scraping data records from HTML files) is difficult to accomplish with a single code snippet.

We propose {\it code ensemble with evolution} to address this challenge. It produces a set of code snippets that compensate for each other. The code ensemble can then effectively handle scenarios that require complex logic in the program. We design an evolutionary method to iteratively optimize the ensemble. Next, we elaborate this technique in more detail.

\begin{figure}[t]
  \vspace{-2mm}
  \centering
  \includegraphics[width=0.95\linewidth]{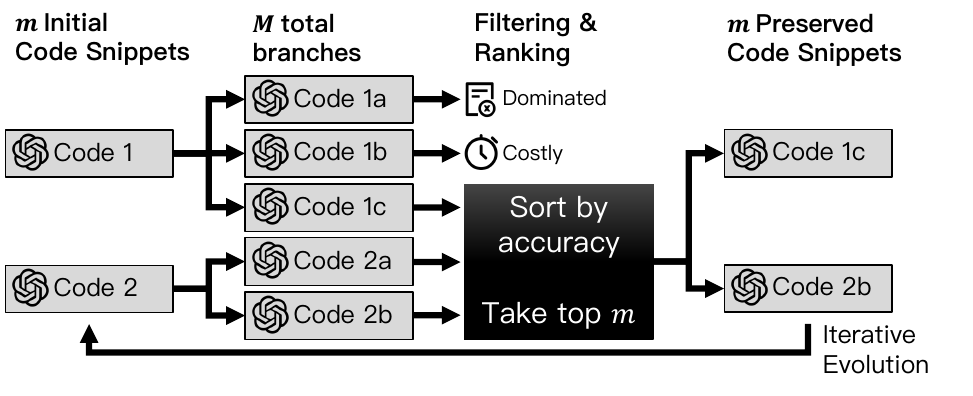}
  \vspace{-6mm}
  \caption{Illustration of one evolution iteration.}
  \label{fig:codegen-ensemble}
  \vspace{-6mm}
\end{figure}

\sys produces multiple modular code snippets and combines them to form a complete solution. In particular, we design an {\it evolutionary algorithm} which produces a diverse set of programs, automatically measures their quality (using generated test cases), and aggregates the most effective ones to jointly solve the task. 

\begin{algorithm}[t!]
    \setstretch{0.60}
    \caption{Code Ensemble Evolution}\label{alg:ensemble}
	\LinesNumbered
    \smaller
	\KwIn{Intial Code Snippets $C^0 = \lbrace R \rbrace_{i=1}^m$}
	\KwOut{Code Ensemble $C^t$}
        \For{$t \leftarrow 0$ to $T-1$}{
            $C^{t+1} \leftarrow C^t$; \\
            \For {each $R \in C^t$} {
                $E \leftarrow \texttt{get\_errors\_in\_verificatiion}(R)$; \\
                \For {each $e \in E$}{
                    $A \leftarrow \texttt{get\_advices}(R, \lbrace e \rbrace)$; \\
                    \For {each $a \in A$}{
                        $C^{t+1}.add(\texttt{fixing}(R, \lbrace e \rbrace, a))$; \\
                    }
                }
                $A \leftarrow \texttt{get\_advices}(R, E)$; \\
                \For {each $a \in A$}{
                    $C^{t+1}.add(\texttt{fixing}(R, E, a))$; \\
                }
            }
            \For {each $R_i \in C^{t+1}$} {
                \If {$R_i$ timeouts or fallbacks too often}{
                    $C^{t+1}.del(R_i)$; \\
                }
            }
            \For {each $R_i \ne R_j \in C^{t+1}$} {
                \If {$R_j$ dominates $R_i$}{
                    $C^{t+1}.del(R_i)$; \\
                }
            }
            sort $C^{t+1}$ by accuracy from high to low; \\
            $C^{t+1} \leftarrow C^{t+1}_{0...m-1}$; \\
        }
        return $C^T$; \\
\end{algorithm}
\noindent\textbf{The Evolutionary Algorithm.}
As shown in Alg.~\ref{alg:ensemble}, the evolutionary algorithm contains two parts: {\it {code branching}} and {\it {code filtering}}. \sys starts by creating $m$ initial code snippets. Then, in each evolutionary iteration, with the aid of LLMs, each code snippet branches into several new versions. All code snippets are then ranked and filtered based on their quality. At the end of each iteration, only the top-$M$ ($M > m$) code snippets are preserved and are ready to branch in the next iteration. This evolutionary process continues until the code ensemble does not get better anymore on the validation set.

\sys uses the following {\it code branching} strategies:
\begin{compactitem}
\item {\textbf {Different Tools}}: Here we ask LLMs to use different tools when generating the initial code snippet. For example, in semantic-related problems, using the $\texttt{nltk}$ tool could lead to a syntax tree-based solution, while using $\texttt{fuzzywuzzy}$ could lead to a string matching-based solution.

\item {\textbf {Prompt Rephrasing}}: Here we use different task descriptions $p_{task}$ when generating the initial code snippet. Rather than leaving the user with the prompt engineering job, using a prompt generator~\cite{awesome-chatgpt-prompts} to rephrase the user-provided task prompt could lead to different advice and potentially different solutions.
    
\item {\textbf {Diversify Advice}}: When generating the initial code snippet or fixing a code snippet, we explicitly ask LLMs to produce multiple different pieces of advice to generate different code snippets (Alg.~\ref{alg:ensemble}, L7, L10). 
    
\item {\textbf {Varied Test Cases}}: When fixing a code snippet, if the verification produces multiple error test cases, \sys presents a subset of errors to LLMs for advice: pick each one of the error cases (Alg.~\ref{alg:ensemble}, L5$\sim$6) as well as the full error set (Alg.~\ref{alg:ensemble}, L9). For a set of $n$ errors, this produces $n+1$ branches.
\end{compactitem}
\vspace{.1in}
If, after applying the branching strategies to each of the initial $m$ code snippets, more than $M$ branches emerge, as in Fig.~\ref{fig:codegen-ensemble}, \sys performs {\it {code filtering}} to rank and eliminate ineffective code snippets. A code snippet will be excluded if it meets one of the following criteria:
\begin{compactitem}
    \item {\textbf {Dominated}}: If another code snippet is correct on every test case where the current snippet is correct, the latter is deemed redundant and thus eliminated (Alg.~\ref{alg:ensemble}, 16).
    
    \item {\textbf {Costly}}: Code snippets are encouraged to fall back to LLM queries when confidence is low. However, if a snippet always relies on LLM queries to produce correct answers, it becomes too expensive to utilize and should therefore be discarded. Similarly, snippets that potentially leverage computationally expensive tools are given lower priority or excluded (Alg.~\ref{alg:ensemble}, L13).
    
    \item {\textbf {Inaccurate}}: Code snippets that are incorrect and produce an excessive number of false positives should be discarded. Specifically, after removing dominated and costly code snippets, the surviving code segments are sorted by accuracy on the validation set from high to low and at most $M$ snippets with the highest accuracy are preserved (Alg.~\ref{alg:ensemble}, L18).
\end{compactitem}

After each iteration, at most $M$ snippets are preserved, where $M > m$ is set as a parameter. This process iterates until convergence. The top-$m$ preserved code snippets from the final iteration are then ensembled into the final solution for inference.

\vspace{.05in}
\noindent\textbf{Ensembling.} The code snippets are ensembled either in {\it parallel} or {\it sequentially}.
By default, the code snippets are executed in parallel, and their execution results are aggregated to produce the final answer. This approach is suitable for most scenarios where the return values of the code snippets can be straightforwardly combined. For example, in classification tasks where the code snippets return values from a set of pre-determined classes, a simple majority voting strategy will work. Alternatively, an accuracy-weighted voting strategy can be employed to provide additional flexibility by assigning different weights to the snippets based on their accuracy on the test data, mitigating the influence of the low-accuracy candidates.

While the parallel ensemble is generally suitable, there are specific scenarios where aggregation is challenging. For instance, in data extraction tasks, where the extracted data can be any string, different snippets typically return different results. Running these snippets in parallel and taking the union of all reported answers may result in poor precision (high false positives), while taking their intersection may lead to poor recall (high false negatives). \sys addresses these cases by using a sequential ensemble.

The sequential ensemble first sorts all code snippets by their accuracy on a validation set, from high to low. Then on each data record, \sys executes the code snippet with the highest accuracy first. If this code snippet fails on this data record and returns `None', the code snippet with the second highest accuracy will be triggered, and so on. If all code snippets fail, the system then falls back to the next module in the pipeline. With this cascading execution strategy, \sys dynamically adapts to the data and runs different code snippets on different data records, thus improving the accuracy.

\subsection{ModelGen Module}
\label{sec:modules-model}

In \sys, the {\it ModelGen} module distills a small machine learning model using the LLM responses as pseudo-ground truth. Although using LLMs for annotation is a known idea~\cite{llmlabel}, choosing the model architecture appropriate for a given \dm task presents a challenge due to the vast number of DNN architectures available and the diversity of \dm tasks. It is often hard for the users to decide on the model architecture fitting their tasks.

\sys does not rely on users to select and tune models. Our key insight here is that almost all tasks with expected input-output pairs $\lbrace (x,y) \rbrace$ can be formulated as Seq2Seq generation tasks~\cite{seq2seq,t5}, if we appropriately serialize $x$ and $y$ into a list of tokens. 
For instance, a data imputation task that infers the age of a person given their birth year and death year can be serialized as $x = \textit{``birth\_year=?,death\_year=?}$ and $y = \textit{``living\_age=?''}$, where {\it{`?'}}s are replaced with the decimal string of the actual year and age data. The Seq2Seq paradigm has been shown to be effective and generalizable in supporting NLP tasks~\cite{bart,t5,glm}, and has led to the development of pre-trained models such as T5~\cite{t5} that are widely applicable to a range of tasks.

This suggests it should be possible for us to employ a Seq2Seq-based architecture that can universally support many \dm tasks. For most \dm tasks, except classification and numeric tasks such as data imputation and data extraction, we can directly leverage the Seq2Seq paradigm, while for classification tasks such as entity resolution and data annotation, only the encoder component of the Seq2Seq model is utilized. The encoder is concatenated with an additional Multi-layer Perceptron (MLP) to function as either a classifier or a regressor, replacing the decoder in the Seq2Seq model. In these cases, rather than outputting a list of tokens, the small model produces either a probability distribution over candidate classes or a numeric value. \sys deduces the task type based on the output type, which is specified in the user configuration. A string type suggests a general Seq2Seq model; a float type suggests a regression; while a categorical type suggest a classification task.

When the ModelGen module is uncertain, if falls back to other modules. Specifically, a confidence score can be derived from the output of the deep neural network, and records with low confidence shall fall back to the LLM module. For Seq2Seq tasks, we use the perplexity $PPL(y) = \exp\left(-\frac{1}{n}\sum_{j=1}^n \log p(y_j \vert x, y_{1...j-1}) \right)$ of the generated text to compute confidence. Intuitively, perplexity is the inverse probability that the model predicts for each token in the output training sequence, given the tokens that came before it. A lower perplexity value corresponds to a higher level of confidence in the prediction. Thus, $\frac{1}{PPL(y)} \in (0, 1]$ could be used as the confidence indicator. As for classifiers, the probability associated with the highest predicted class can serve as an indicator of confidence. The closer this probability is to 1, the higher the confidence in the prediction. More specifically, given a $K$-way classification prediction, which is a probability distribution $\lbrace p_y^i \rbrace_{i=1}^K$, where $p_y^i$ is the presumed probability that the instance belongs to class $i$, then $\frac{K \cdot \max_i(p_y^i)-1}{K-1} \in [0,1]$ could be used as the confidence. Such a confidence threshold is a hyperparameter searched by the optimizer.

\end{sloppypar}
\section{Infrastructure}
\label{sec:infra}
\begin{sloppypar}

\sys infrastructure features {\it cache storage}, {\it query batching}, and {\it tools integration}. Universally supporting various \dm tasks, these components boost the performance of the automatically generated, domain specific \dm solution.

\vspace{-2mm}
\subsection{Cache Storage}
\label{sec:infra-cache}
The cache storage preserves all LLM queries and their corresponding responses throughout the compilation and execution phases. ModelGen utilizes the cached LLM responses as training data, while CacheReuse constructs a vector index and directly employs the indexed responses to respond to new LLM queries. Additionally, \sys logs the LLM calls invoked by the CodeGen module during code generation, thereby reducing the expense of repetitive code validation. Apart from storing LLM queries, \sys also caches the input-output pairs of each module, enabling cached validation when executing the specialized \sys optimizer (Sec.~\ref{sec:optim-specialized}).

\subsection{Query Batching}
\label{sec:infra-batching}

Query batching merges multiple queries into a single batched query. For the CacheReuse and ModelGen modules, it allows for batched vector computation. For the CodeGen module, it allows potential parallel execution. But more importantly, for the LLM modules, it eliminates the redundancy among the task descriptions and other auxiliary prompts, thus reducing the overall number of tokens sent to LLMs. Furthermore, in the context of few-shot or zero-shot learning, batching queries supplies LLMs with additional examples. These examples carry more information with respect to the data distribution, making LLMs more robust and reliable in answering queries, thereby potentially improving the effectiveness (Tab.~\ref{tab:topic-classification}).

For example, in a data annotation task that classifies the text into different topics, instead of asking separately {\it {``Here are some topics: A. Sports B. Business C. Education D. Gaming ... What is the topic of the following paragraph? $\texttt{<TEXT>}$''}} for each instance $\texttt{<TEXT>}$, we could ask {\it {``Here are some topics: A. Sports B. Business C. Education D. Gaming ... What is the topic of each of the following paragraphs? 1.$\texttt{<TEXT1>}$ 2.$\texttt{<TEXT2>}$ 3.$\texttt{<TEXT3>}$ ...''}}.An LLM is then instructed to reply in the format of {\it {``1. A. 2. F. 3. B. ...''}}.

\sys investigates 5 methods to form a batch. These methods use the embedding of each individual data record produced by an Sentence-BERT encoder. Given $n$ queries and a fixed batch size $B$:
\begin{compactitem}
\item $\texttt{RND}$: random batching. It randomly composes $n/B$ batches, making each batch adhere to the distribution of the query workload.

\item $\texttt{DIV}$: diverse batching. It runs $B$-way balanced clustering on the embeddings, and then assigns queries in different clusters to the same batch.

\item $\texttt{PRX}$: proximal batching. It runs $n/B$-way balanced clustering on the embeddings, and the queries in the same cluster form a batch.

\item $\texttt{FAR}$ \& $\texttt{CLS}$: distance-based batching. It randomly samples a starting point and then iteratively adds the data point farthest/nearest from existing queries into the batch.
\end{compactitem}

Intuitively, $\texttt{SIM}$ and $\texttt{CLS}$ present similar queries jointly to LLMs. The intuition underneath is to allow LLMs to see similar cases and thus easier to make decisions. On the other hand, $\texttt{DIV}$ and $\texttt{FAR}$ present diverse queries together to LLMs. This allows LLMs to see diverse cases and make more robust decisions. 

\subsection{Tools Integration}
\label{sec:infra-tools}

\sys seamless integrates the tools supplied by the users into the solution such that it could efficiently access local data and extract the required information on demand, without having to upload the entire data to LLMs. As an example, in a table retrieval task, allowing LLMs to use some tools to collect information from the local databases will help them better align the user queries with discovered tables and thus significantly enhance the correctness.

Once the users specify a set of tools, \sys has to decide which tools to use for a given LLM query. This is challenging because different queries, even if they belong to the same task, might need different tools under different situations. Moreover, it might need multiple tools to collaboratively solve a complex problem, where the order of calling these tools matters. 

\sys leverages the reasoning ability of LLMs to address the above challenges. Rather than count on one single LLM query to solve the problem at once, \sys interacts with LLMs, automatically decomposes the original complex problem into multiple steps, and solves the problem step by step. In this iterative task-solving process, it automatically and dynamically decides which tool to use based on the feedback acquired in the last step.

\begin{figure}[t]
  \centering
  \vspace{-2mm}
  \includegraphics[width=0.8\linewidth]{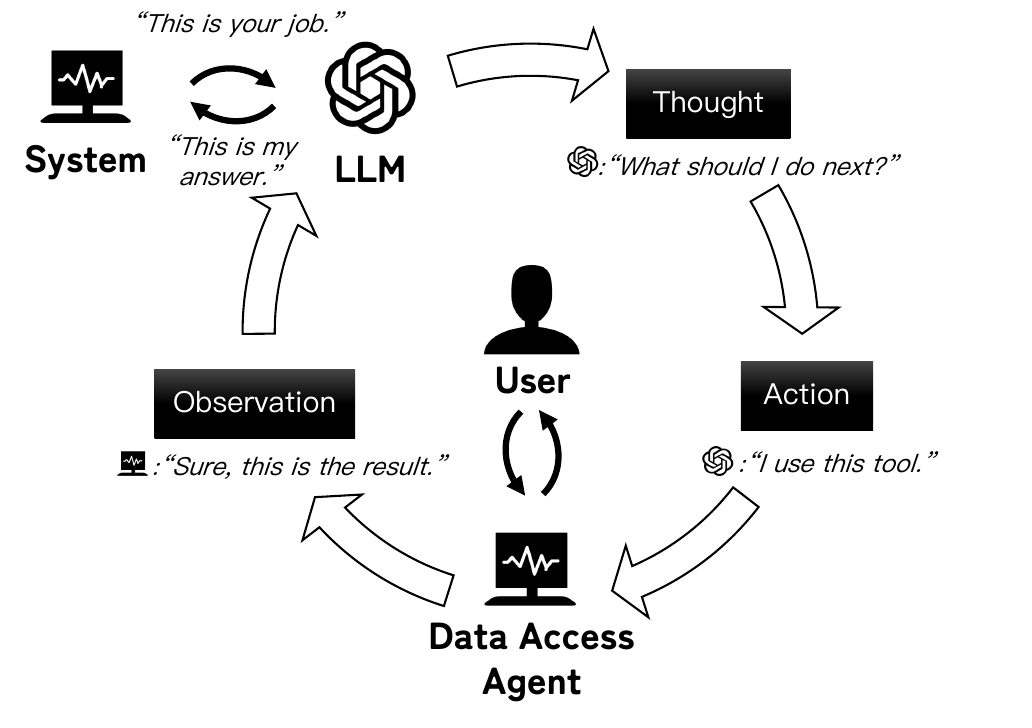}
  \vspace{-4mm}
  \caption{Illustration of iterative tool invocation.}
  \label{fig:tools}
  \vspace{-4mm}
\end{figure}

As shown in Fig.~\ref{fig:tools}, there are three key parts in this process: action, thought, and observation~\cite{react}. Each action corresponds to a tool to be used. A thought corresponds to an idea that LLM suggests taking to solve the problem, while an observation represents the execution results after running a tool.

More specifically, given an input $x$ and a list of actions $A = (T \bigcup \lbrace \texttt{SUBMIT} \rbrace) \times S$, where $T$ represents the set of tools and $S$ denotes their parameters. The $\texttt{SUBMIT}$ action is a special tool that terminates the multi-step iterative process. In each step, the LLM first outputs its thought $p_i$ in natural language, which triggers one of the actions $a_i = (t_i,s_i) \in A$. Consequently, \sys would invoke the tool $t_i$ with parameter $s_i$ corresponding to action $a_i$, resulting in an observation $o_i = t_i(s_i)$. Subsequently, based on the newly received observation $o_i$, the LLM produces a new thought $p_{i+1}$ which invokes a new action $a_{i+1} \in A$. This process continues until an observation occurs that triggers a special termination tool $a = (\texttt{SUBMIT}, y)$, where the parameter $y$ represents the LLM's final answer.

We use a table retrieval task as an illustrative example to further elucidate this process. With a set of pre-defined tools, given a query, the LLM starts by generating a thought that it should first get candidates through $\texttt{SEARCH\_KEYWORDS}$ or $\texttt{BM25}$. After the LLM chooses an action, \sys will invoke the corresponding tool and returns a list of candidate tables. Based on the list of candidate tables, the LLM may get a new thought that a few tables are worth further investigation. Consequently, it selects the $\texttt{GET\_SCHEMA}$ action to get the schema so that it has more information of these tables. After \sys returns the observation, the LLM either proceeds to get more candidates or invokes the $\texttt{SUBMIT}$ action to submit some of these tables as the final answer.
\end{sloppypar}
\section{Experiments}
\label{sec:exp}

Our experiments focus on answering the following questions:
\begin{compactitem}
\item How do the solutions that \sys produces perform in comparison to its generic counterpart and other task specific solutions?

\item How do the solutions that \sys produces perform in comparison to using LLMs to process each data record?

\item how does the \sys optimizer perform?

\item How do each module in \sys and the respective optimizations we propose perform?

\end{compactitem}

As shown in Tab.~\ref{tab:exp-setup}, we selected $9$ datasets spanning $5$ tasks to show the performance of \sys. Together these tasks cover all types of modules in \sys and the typical execution pipelines. We present the experiments {\it task by task}, including the datasets the task runs on, the evaluation metric, the baseline methods we compare against, and the results. We evaluate the internals of the \sys optimizer on the entity resolution task with 2 datasets. Due to the limited space, we describe the user configuration (input text) in Appendix~\ref{sec:appendix:config} and the data discovery experiment in Appendix~\ref{sec:exp:dd} of the extended version~\cite{seed}. 

\noindent\textbf{Parameter Settings.}
For all the experiments, we employ a GPT-4 \cite{gpt4} with temperature $T=0$ as the LLM backend. Since \sys primarily operates through code and small models, or calls the LLM service, there is no special hardware required to run the experiments. In the case of ModelGen, we initialize it using a \texttt{t5-small} model and then fine-tune them on $1024$ unlabeled validation data, employing LLM-generated pseudo labels to simulate a deployed system. The CacheReuse model uses a pre-trained frozen \texttt{t5-small} model checkpoint. CodeGen selects its branching number from either $m=1$ or $m=4$, where a branching number of $m=4$ means that four code snippets are generated. As the cost of calling LLMs dominates the cost of the execution pipeline, we use the number of queries processed by the LLM to represent the cost $\mathcal{C}$ of a plan, denoted as ``LLM R''.

\begin{table}[htbp]
\scriptsize
\centering
  \caption{Task, Datasets, and utilized components.}
  \vspace{-4mm}
  \label{tab:exp-setup}
  \begin{tabularx}{1.0\linewidth}{llcccc}
    \toprule
    Task                                         & Dataset                   & CodeGen                      & CacheReuse                    & ModelGen                      & LLM                           \\
    \toprule
    \multirow{2}{*}{\scriptsize Entity Resolution }   & {\scriptsize Amazon-Google }   & \multirow{2}{*}{\checkmark}  & \multirow{2}{*}{\checkmark}   & \multirow{2}{*}{\checkmark}   & \multirow{2}{*}{\checkmark}   \\
                                                 & {\scriptsize Abt-Buy       }   &                              &                               &                               &                               \\
    \midrule
    \multirow{2}{*}{\scriptsize Data Imputation   }   & {\scriptsize Buy           }   & \multirow{2}{*}{\checkmark}  & \multirow{2}{*}{\checkmark}   & \multirow{2}{*}{\checkmark}   & \multirow{2}{*}{\checkmark}   \\
                                                 & {\scriptsize Restaurant    }   &                              &                               &                               &                               \\
    \midrule
    \multirow{1}{*}{\scriptsize Data Extraction   }   & {\scriptsize Wiki NBA      }   & \multirow{1}{*}{\checkmark}  &                               &                               & \multirow{1}{*}{\checkmark}   \\
    \midrule
    \multirow{1}{*}{\scriptsize Data Annotation   }   & {\scriptsize OS            }   &                              & \multirow{1}{*}{\checkmark}   &                               & \multirow{1}{*}{\checkmark}   \\
    \midrule
    \multirow{3}{*}{\scriptsize Data Discovery    }   & {\scriptsize Spider        }   &                              &                               &                               & \multirow{3}{*}{\checkmark}   \\
                                                 & {\scriptsize OTT-QA        }   &                              &                               &                               &                               \\
                                                 & {\scriptsize KGDBQA        }   &                              &                               &                               &                               \\
    \bottomrule
  \end{tabularx}
  \vspace{-4mm}
\end{table}

\subsection{Entity Resolution}
\label{sec:exp:er}
Entity resolution determines whether two data records refer to the same entity.

\noindent\textbf{Datasets.} We use the {\it Amazon-Google} and {\it Abt-Buy} datasets~\cite{ee-benchmark} for evaluation. Both tasks aim to determine whether two online products are the same. Amazon-Google has three attributes: title, manufacturer, and price. Abt-Buy has three attributes as well: name, description, and price. We adopt the test split used in MatchGPT~\cite{matchgpt}.

\noindent\textbf{Evaluation Metric.} We report the {\it F1} score, as in previous works.

\noindent\textbf{Baselines.} We compare against two approaches: (1) DITTO~\cite{ditto}, a supervised learning method trained on thousands of labeled examples. (2) MatchGPT~\cite{matchgpt}, the current SOTA in using LLMs to solve the entity resolution problem. MatchGPT has two versions. MatchGPT Generic manually selects few-shot examples, while MatchGPT Specialized utilizes the Jaccard similarity between two entities to retrieve few-shot examples. 
\begin{table}[htbp]
\small
\centering
\vspace{-2mm}
\caption{Entity Resolution: Quantitative Results}
\vspace{-4mm}
\label{tab:entity-resolution}
    \begin{tabularx}{0.94\linewidth}{lcrcr}
        \toprule
        \multirow{2}{*}{Method}     & \multicolumn{2}{c}{Amazon-Google} & \multicolumn{2}{c}{Abt-Buy}   \\
                                    & F1        & LLM R                 & F1            & LLM R         \\
        \midrule
        DITTO                       & $80.7$    & N / A                 & $91.3$        & N / A         \\
        MatchGPT Generic            & $76.4$    & $100.0\%$             & $95.8$        & $100.0\%$     \\
        MatchGPT Specialized        & $85.2$    & $100.0\%$             & $94.4$        & $100.0\%$     \\
        \sys LLM Only               & $77.4$    & $100.0\%$             & $92.1$        & $100.0\%$     \\
        \sys Max                    & $79.2$    &  $22.4\%$             & $92.8$        &  $27.1\%$     \\
        \sys Opt                    & $76.2$    &  $10.4\%$             & $88.0$        &   $2.9\%$     \\
        \bottomrule
    \end{tabularx}
\vspace{-2mm}
\end{table}

In Table~\ref{tab:entity-resolution}, ``LLM Only'' represents a solution that uses only the LLM module in \sys. ``\sys Max'' represents the best F1 score among all the plans enumerated by the \sys optimizer. ``\sys Opt'' denotes the plan selected by the \sys optimizer, when the effectiveness gap parameter $G$ (Def.~\ref{thm:seed-optim-def}) is set to 5\%.  Table~\ref{tab:entity-resolution} shows that the \sys optimizer produces a plan that reduces the number of LLM calls by at least $90\%$, while guaranteeing that the F1 is at most 5\% lower than \sys Max. This confirms the effectiveness of the \sys optimizer. 
The accuracy of \sys is close to DITTO trained on thousands of labels and outperform MatchGPT Generic. MatchGPT Specialized has a higher F1 than \sys Max, because it uses prompts carefully tuned w.r.t. the dataset. \sys instead uses a generic prompt template without prompt engineering.

\begin{table}[htbp]
\tiny
\centering
\vspace{-2mm}
\caption{Entity Resolution: \sys Optimizer}
\vspace{-4mm}
\label{tab:entity-resolution-optim}
    \begin{tabularx}{0.98\linewidth}{lcccccccccc}
        \toprule
            & \multicolumn{5}{c}{Amazon-Google}                                             & \multicolumn{5}{c}{Abt-Buy} \\
            \cmidrule(lr){2-6}                                                              \cmidrule(lr){7-11}
            & \multicolumn{2}{c}{Max} & \multicolumn{2}{c}{Opt} & \multirow{2}{*}{\#Plan}   & \multicolumn{2}{c}{Max} & \multicolumn{2}{c}{Opt} & \multirow{2}{*}{\#Plan}   \\
            \cmidrule(lr){2-3}        \cmidrule(lr){4-5}                                    \cmidrule(lr){7-8}        \cmidrule(lr){9-10}
            & F1     & LLM R          & F1     & LLM R                &                     & F1     & LLM R          & F1     & LLM R          &                           \\
        \midrule
        BF  & $79.2$ & $22.4\%$       & $76.8$ & $ 3.7\%$             & $50146$             & $92.8$ & $27.1\%$       & $87.9$ & $1.8\%$        & $50146$                   \\
        GEN & $79.2$ & $22.4\%$       & $76.8$ & $ 3.7\%$             & $19541$             & $92.8$ & $27.1\%$       & $87.9$ & $1.8\%$        & $20300$                   \\
        SPE & $79.2$ & $22.4\%$       & $76.2$ & $10.4\%$             & $  650$             & $92.8$ & $27.1\%$       & $88.0$ & $2.9\%$        & $  782$                   \\
        \bottomrule
    \end{tabularx}
\vspace{-4mm}
\end{table}

Next, we evaluate the effectiveness of the \sys optimizer. Table~\ref{tab:entity-resolution-optim} presents the results of three different optimization approaches: `BF' represents an exhaustive search over all possible plans; GEN represents the generic \sys optimizer (Section~\ref{sec:optim-general}); while SPE represents the specialized \sys optimizer (Section~\ref{sec:optim-specialized}). Note that the generic \sys optimizer consistently identifies the same ``Max'' and ``Opt'' plans to the exhaustive search, indicating that the assumption in Sec.~\ref{sec:optim} that the descendants or a dominated plan would not become a part of the final optimal plan holds in practice. The specialized \sys optimizer SPE produces the same ``Max'' plan but a slightly less efficient ``Opt'' plan. However, it reduces the number of plans that need to be searched by $99.5\%$. In the ``Opt'' plan produced by SPE, the LLM only handles $10.4\%$ of the queries, while CacheReuse and ModelGen answer $41.7\%$ and $47.9\%$ of the queries respectively.

\subsection{Data Imputation}
\label{sec:exp:di}
As a typical data cleaning task, data imputation~\cite{imputation} replaces missing or corrupted data with substituted values.

\noindent\textbf{Datasets.} We use the {\it Buy} dataset~\cite{fm} and {\it Restaurant} dataset~\cite{fm} for evaluation. The Buy dataset has three attributes: product name, product description, and manufacturer. In this dataset, the manufacturer is masked and considered as the missing attribute to be imputed. The Restaurant dataset has five attributes: restaurant name, address, city, phone number, and food type. In this dataset, the city is masked and considered as the missing attribute to be imputed. For both datasets, when using CodeGen, we randomly select $3$ examples that make at least one initially generated code snippets fail in validation.

\noindent\textbf{Evaluation Metric.} We report the {\it imputation accuracy}, where an imputation is considered correct if it exactly matches the gold value, non-trivially contains the gold value, or is non-trivially contained within the gold value.

\noindent\textbf{Baselines.} The only available baseline for few-shot imputation on the Buy dataset is FMs~\cite{fm}, a pure end-to-end LLM solution with $10$-shot examples. In addition, we compare against the statistical data cleaning method HoloClean~\cite{holoclean} and the supervised method IMP~\cite{imp} which uses a large number of training data.

\begin{table}[htbp]
  \caption{Data Imputation}
  \vspace{-4mm}
  \label{tab:data-imputation}
  \begin{tabularx}{0.95\linewidth}{lcccc}
    \toprule
    \multirow{2}{*}{Method}             & \multicolumn{2}{c}{Buy}   & \multicolumn{2}{c}{Restaurant}\\
                                        & Acc       & LLM R         & Acc           & LLM R         \\
    \midrule
    HoloClean~\cite{holoclean}          & $16.2$    & N / A         & $33.1$        & N / A         \\
    IMP~\cite{imp}                      & $96.5$    & N / A         & $77.2$        & N / A         \\
    FMs~\cite{fm}                       & $98.5$    & $100.0\%$     & $88.4$        & $100.0\%$     \\
    \sys LLM only (Max)                       & $96.9$    & $100.0\%$     & $89.8$        & $100.0\%$     \\
    \sys Opt                               & $91.9$    &  $39.3\%$     & $84.7$        &   $0.0\%$     \\
    \bottomrule
  \end{tabularx}
  \vspace{-2mm}
\end{table}

As indicated in Table~\ref{tab:data-imputation}, \sys achieves an impressive imputation accuracy of $96.1\%$ on the Buy dataset, using only $3$ examples. This performance significantly surpasses the generic data cleaning method HoloClean~\cite{holoclean}. Furthermore, it demonstrates comparable results to the state-of-the-art supervised solution IMP~\cite{imp}, which relies on training with thousands of labeled examples. When compared to FMs~\cite{fm}, which employ the LLM on each row in the table, \sys Opt represents the plan produced by the specialized \sys optimizer, with the accuracy gap $G$ set to 5\%. LLM only is the `\sys Max' plan that the \sys optimizer finds. \sys Opt reduces the number of LLM calls by $60.7\%$ on the Buy dataset, with only a $6.6\%$ accuracy gap compared to the supervised IMP. In this plan, only $39.3\%$ of queries are answered by the LLM, while $55.3\%$ are answered by CodeGen, and $5.5\%$ are answered by ModelGen. 
On the Restaurant dataset, \sys Opt uses only the CodeGen module, without making any LLM calls at all; and its accuracy is only 5\% lower than the LLM only solutions.

For the code generated by \sys, please refer to Appendix~\ref{sec:appendix:code} of the extended version~\cite{seed}.

\subsection{Data Extraction}
\label{sec:exp:ie}
This task extracts structured information from unstructured data~\cite{ie}.

\noindent\textbf{Datasets.} We use the {\it Wiki NBA} dataset~\cite{evaporate}, which contains Wikipedia pages of NBA players. Each page is in a complex HTML format. This task extracts $19$ user-defined attributes for each player, including name, height, college attended, and other relevant facts. When using CodeGen, for each of the $19$ attributes to be extracted, we provided one example to help code generation: since there are empty attributes, we randomly select one example that has a non-empty value on this attribute.

\noindent\textbf{Evaluation Metric.} As in EVAPORATE~\cite{evaporate}, we report {\it token F1}. This is the geometric mean of  the precision and recall, where precision measures the ratio of matched tokens among all extracted tokens and recall measures the ratio of matched tokens among all correct tokens.

\noindent\textbf{Baselines.} We compare against pre-trained machine learning models SimpleDOM, RoBERTa-Base, RoBERTa-Structural, and DOM-LM~\cite{dom}; EVAPORATE-DIRECT~\cite{evaporate}, which is a pure LLM method; and EVAPORATE-CODE+~\cite{evaporate} which is an LLM-based code generation approach. Note the EVAPORATE approaches are specifically designed for data extraction.

\begin{table}[htbp]
  \vspace{-2mm}
  \caption{Data Extraction: Quantitative Results}
  \vspace{-3mm}
  \label{tab:data-extraction}
  \begin{tabularx}{0.7\linewidth}{lc}
    \toprule
    Method                              & Wiki NBA F1 \\
    \midrule
    SimpleDOM~\cite{simple-dom}         & $28.8$    \\
    RoBERTa-Base                        & $34.2$    \\
    RoBERTa-Structural                  & $48.6$    \\
    DOM-LM~\cite{dom}                   & $64.8$    \\
    EVAPORATE-DIRECT~\cite{evaporate}   & $84.6$    \\
    EVAPORATE-CODE+~\cite{evaporate}    & $84.7$    \\
    \sys Opt                               & $88.6$    \\
    \bottomrule
  \end{tabularx}
  \vspace{-5mm}
\end{table}

As shown in Table~\ref{tab:data-extraction}, the plan produced by the \sys optimizer achieves the highest F1 score among all methods, with only $22.5\%$ cases processed by the LLM. The Performance of the trained machine learning models in general is poor, confirming that small models are lack of the semantics understanding ability to accurately extract data.
However, EVAPORATE, which uses LLM to either directly extract data or produce code, shows impressive performance, indicating that using LLMs to synthesize code is effective on this task.
For the code generated by \sys please refer to Appendix~\ref{sec:appendix:code} of the extended version~\cite{seed}.

\subsection{Data Annotation}
\label{sec:exp:cls}
Data annotation classifies data records into different categories.

\noindent\textbf{Datasets.} We use the {\it OS} dataset~\cite{os} for evaluation. This is a binary classification tasks where the goal is to annotate a tweet as either offensive or benign.

\noindent\textbf{Evaluation Metric.} We measure the {\it classification accuracy} as used in EFL~\cite{efl}.

\noindent\textbf{Baselines.} We compare against EFL~\cite{efl} which achieves state-of-the-art few-shot results on the OS dataset amongst non-LLM approaches. We also compare against FT, which is a few-shot fine-tuning of a language model as the basic approach. Both EFL and FT use $8$-shot examples. We also compare with an LLM-only approach.

\begin{table}[htbp]
  \vspace{-2mm}
  \caption{Data Annotation: Quantitative Results}
  \vspace{-4mm}
  \label{tab:topic-classification}
  \begin{tabularx}{0.95\linewidth}{lcrr}
    \toprule
    Method                  & Accuracy          & LLM R     & \#LLM     \\
    \midrule
    FT~\cite{efl}           & $70.0$            & N / A     & N / A     \\
    EFL~\cite{efl}          & $79.8$            & N / A     & N / A     \\
    \sys LLM Only           & $84.2$            & $100.0\%$ & $100.0\%$ \\
    \sys Opt              & $88.3$            &  $31.1\%$ &   $1.0\%$ \\
    \bottomrule
  \end{tabularx}
  \vspace{-2mm}
\end{table}

As shown in Tab.~\ref{tab:topic-classification}, \sys achieves the best accuracy, while only calling the LLM on 31.1\% of the data records. Moreover, due to the query batching optimization, \sys reduces the number of LLM calls and consumed tokens by $99\%$ and $95\%$ respectively, compared to the approaches that perform one individual LLM call on each data record.  
Note that the accuracy is even better than always querying the LLM. This is because the pseudo-annotated examples accumulated in the model module provide more information about the query workload, overcoming the few-shot limitations of querying the LLM, despite the relatively limited capacity of a small model.
\section{Related Work}
\label{sec:related_work}
\begin{sloppypar}

\noindent\textbf{LLMs in Data Curation.}
Most recently, the \dm researchers have begun to exploit LLMs' abilities to solve \dm problems. Specifically, most studies attempt to directly employ LLMs to solve a task by describing the task in a natural language. Zero-shot or few-shot LLMs are then utilized to generate textual responses, which are converted to the expected task output. For instance, in entity resolution, an LLM is asked if two records in a sentence belong to the same entity~\cite{amazon}. Such empirical research~\cite {gpt3,palm,structgpt,fm} has demonstrated the feasibility of applying LLMs to \dm tasks. However, these studies mainly focus on designing appropriate prompts to make LLM effective for a specific task, while we aim to offer a generic tool to automatically synthesize domain-specific \dm solutions.

Directly using LLMs as a universal solution from \dm tasks, despite its convenience, has very high performance overhead as we have shown.
Although some early studies have targeted some of these problems caused by the limited tunability and high computational cost of LLMs, these works, again, only focus on individual \dm tasks. For instance, EVAPORATE~\cite{evaporate} focuses on information extraction; Chameleon~\cite{chameleon} targets question answering; PromptNER~\cite{prompt-ner} focuses on named entity recognition, etc.

\noindent\textbf{Applied, Open Source LLM Projects.}
Apart from academic research, large open-source projects like langchain~\cite{langchain} have emerged to integrate recent advances in large models into a single platform. Nevertheless, similar to other infrastructures like hugginggface~\cite{huggingface}, langchain mainly serves as an implementation platform where the users can leverage the existing techniques to develop LLM-based applications. It does not provide any built-in support to better solve \dm problems. Moreover, the efficiency and scalability issues are largely overlooked.

Individual applied projects are emerging as well, which implement and integrate practical optimizations into LLMs. For example, GPTCache~\cite{gpt-cache} caches LLM QA responses for reuse; LlamaIndex~\cite{llamaindex} develops retrieval-augmented LLMs to overcome the hallucination problem; AgentGPT~\cite{agent-gpt} focuses on customizing LLMs as dialogue agents; Auto-GPT~\cite{autogpt} includes textual memory and internet access to search and gather information; ChatDB~\cite{chatdb} extends LLMs with external symbolic memories, etc. However, each project only focuses on one specific optimization with straightforward implementation, rather than systematically studying and addressing the challenges raised in effectively and efficiently using LLMs to construct \dm solutions.


\noindent\textbf{LLMs with Tools.}
Since ChatGPT plug-ins~\cite{chatgpt}, a substantial number of studies have investigated the combination of LLMs and tools. For example, HuggingGPT~\cite{hugginggpt} uses LLM to choose huggingface checkpoints; Chameleon~\cite{chameleon} combines LLMs with a set of tools including web search and program execution; StructGPT~\cite{structgpt} uses LLMs to generate parameters for function calling. Inspired by ReAct~\cite{react}, the tools invocation in \sys incorporates a cognitive process to chose actions and provide feedback after each action. Moreover, \sys offers a set of predefined tools specifically designed for \dm tasks and allows users to use the synthesized code and small model modules as new tools.

\end{sloppypar}
\section{Conclusion}
\label{sec:conclusion}
\begin{sloppypar}
In this paper, we introduce \sys, a system that leverages LLMs to synthesize domain-specific \dm solutions, which by fully leveraging LLMs' synthesis, reasoning, semantics understanding abilities as well as the encoded common knowledge, ensure the effectiveness and efficiency of the automatically produced solutions. The results confirm that \sys achieves SOTA performance on various tasks and significantly reduces the number of LLM calls.
\end{sloppypar}

\bibliographystyle{ACM-Reference-Format}
\bibliography{main}
\appendix
\begin{sloppypar}
    
\section{Task Configuration}
\label{sec:appendix:config}
\noindent\textbf{Data Imputation.}
For the Buy dataset, the API is defined as {\it{``data\_imputation(name,desc) -> manufacturer''}}, with description {\it{``name: str. The name of the product.''}} and {\it{``desc: str. The description of the product. Might be empty (string 'nan').''}}. The task is described as {\it{``Given a product's information online, please deduce its manufacturer.''}}. For the Restaurant dataset, the API is defined as {\it{``data\_imputation(name,desc) -> manufacturer''}}, with description {\it{``name: str. The name of the restaurant.''}}, {\it{``addr: str. The address of the restaurant.''}}, {\it{``phone: str. The phone number of the restaurant.''}}, {\it{``type: str. The food type of the restaurant.''}}, and {\it{``city: str. The city the restaurant is in.''}}. The task is described as {\it{``Given a restaurant's information, deduce the city it is located in.''}}. 

\noindent\textbf{Data Extraction.}
For this data extraction task on {\it Wiki NBA} data set~\cite{evaporate}, as $19$ different attributes are required to be extracted, which subject to different rules, we generate $19$ different programs using the attribute placeholder {\it{``<ATTR>''}}. The API is defined as {\it{``information\_extraction\_<ATTR>(html) -> <ATTR>''}}, with description {\it{``html: str. The HTML string.''}} and {\it{``str. The <ATTR> information extracted.''}}. The task is described as {\it{``Given the HTML string of the wikipedia page of an NBA basketball player, extract the player's <ATTR> information.''}}. 

\noindent\textbf{Data Annotation.}
For this task, the API is defined as {\it{``topic\_classification(text) -> topic''}}, with description {\it{``text: str. The tweet message.''}} and {\it{``topic: int. Output 0 for hatespeech, 1 for benign.''}}. The task is described as {\it{``Given a tweet message, determine whether it is hatespeech or benign.''}}.

\noindent\textbf{Entity Resolution.}
For the entity resolution task, both data sets ({\it Amazon-Google} and {\it Abt-Buy}~\cite{ee-benchmark}) share the same API definition {\it{``entity\_resoluion(entity1, entity2) -> is\_same''}} and task description {\it{``Given two products, determine whether they are the same product.''}}, however with different descriptions. The entity description for Amazon-Google is {\it{``entity1: dict. It contains three attributes: `title`, `manufacturer`, `price`. `title` and `manufacturer` are strings, `price` is float.''}}, while the entity description for Abt-Buy is {\it{``entity1: dict. It contains three attributes: `name`, `description`, `price`. `name` and `description` are strings, `price` is float.''}}.

\noindent\textbf{Data Discovery.}
For this task, on all data sets ({\it Spider}~\cite{spider} and {\it OTT-QA}~\cite{ott}, and {\it KaggleDBQA})~\cite{kaggledbqa} datasets.
All datasets share the same task of table retrieval. The difference is that multiple tables in the Spider dataset are related to the query, while OTT-QA has only one ground truth table for each query. For the Spider dataset, the API is defined as {\it{``data\_discovery(query) -> tables\_list''}}, with description {\it{``query: str. The natural language query.''}} and {\it{``tables\_list: List[str]. A list of table names that are found related to the given query.''}}. The task is described as {\it{``Given a natural language query, find tables that can help answer the query.''}}. For the OTT-QA dataset, instead of {\it{``tables\_list''}}, only a single {\it{``table''}} is asked for. In the current infrastructure, we provided five tools: \texttt{GET\_SCHEMA(table\_name)} which gets the schema of a specified table, \texttt{SEARCH\_KEYWORDS(keywords)} which finds the top-20 tables related to a list of keywords using exact string match, \texttt{SEARCH\_VALUE(value)} which finds the top-20 tables related to a value using a fixed $0.2$ Levenshtein distance filter, \texttt{JOINT\_SEARCH(keywords, value)} which retrieves both, and \texttt{BM25(query)} which finds the top-20 tables sorted by BM25.

\section{Experiments}
\label{sec:appendix:ablation}

\subsection{Data Discovery}
\label{sec:exp:dd}
In this experiment, we use a {\it table retrieval} task to evaluate \sys's performance using LLM module with tools integration. As a specific data discovery task, table retrieval finds data that is relevant to a given natural language query.

\noindent\textbf{Datasets.} We use the {\it Spider} dataset~\cite{spider}, the {\it OTT-QA} dataset~\cite{ott}, and the {\it KaggleDBQA} dataset~\cite{kaggledbqa}. Because Spider and KaggleDBQA are NL2SQL benchmarks, we treat one {\it{``question''}} as the input query and the tables involved in the corresponding ground-truth SQL query as the related tables. For the OTT-QA dataset, we use its specified table retrieval task. We use only one manual example query to demonstrate to the LLM how to use the tools. Apart from $\texttt{SUBMIT}$, the available tools include $\texttt{GET\_SCHEMA}$, $\texttt{BM25}$, $\texttt{SEARCH\_KEYWORDS}$, $\texttt{SEARCH\_VALUE}$, and $\texttt{JOINT\_SEARCH}$.

\noindent\textbf{Evaluation Metric.} We measure the {\it F1 score} which reflects both precision and recall, where precision represents the ratio of ground truth tables within discovered tables, and recall measures the ratio of discovered tables within ground truth tables.

\noindent\textbf{Baselines.} We compare \sys with BM25~\cite{INR-019} (a classic information retrieval method) as well as DrQA~\cite{drqa}, an embedding-based method. We consider the top-$3$ results returned by the BM25 and DrQA method as predictions on the Spider and KaggleDBQA dataset. For the OTT-QA dataset, since there is only one ground truth table for each query, we only consider the top-$1$ result returned by the BM25 and DrQA method. Additionally, we compare our solution with a baseline approach that is LLM only, but not augmented with our tool integration. In this case, we present to the LLM the top-$20$ candidates returned by BM25 by concatenating their titles and schema. This is because each time LLM can only ingest at most 20 candidates due to the context length constraint, while iteratively feeding all tables to the LLM is not feasible due to cost issues. 

\begin{table}[htbp]
  \vspace{-2mm}
  \caption{Data Discovery: Quantitative Results}
  \vspace{-4mm}
  \label{tab:discovery}
  \begin{tabularx}{0.85\linewidth}{lccc}
    \toprule
    Method                      & Spider        & OTT-QA    & KGDBQA    \\
    \midrule
    BM25                        & $34.8$        & $62.0$    & $24.1$    \\
    DrQA                        & $35.5$        & $50.0$    & $34.7$    \\
    \sys LLM only               & $50.7$        & $71.0$    & $56.1$    \\
    \sys Opt                    & $85.7$        & $67.8$    & $63.5$    \\
    \bottomrule
  \end{tabularx}
  \vspace{-3mm}
\end{table}

As show in Table~\ref{tab:discovery}, \sys LLM only represents using the LLM module without tools integration, while Opt adds tools. With the assistance of the tools, \sys LLM module surpasses BM25 and DrQA on all datasets and significantly outperforms the solution that only uses the LLM on Spider and KGDBQA. However, \sys is less effective on the OTT-QA dataset. This is because the titles of the tables in this dataset are typically meaningless, thus often misleading \sys when selecting the tools.

\subsection{Ablation Study}
We mainly conducted the ablation study on data annotation on the OS dataset, to investigate the impact of different query batching sizes and strategies, as well as different distance thresholds for the CacheReuse module.

\noindent\textbf{Query Batching.} We tested how the {\it batching} optimization for the LLM module performs. Here \sys LLM Only refers to directly querying GPT-4 with the few-shot prompt produced by \sys's prompt template.

\begin{table}[htbp]
  \vspace{-2mm}
  \caption{Data Annotation: Varying Batch Sizes}
  \vspace{-4mm}
  \label{tab:topic-classification-batching}
  \begin{tabularx}{0.75\linewidth}{lcrr}
    \toprule
    Method                          & OS        & $\#$LLM   & $\#$Tokens    \\
    \midrule
    \sys LLM Only                   & $84.2$    & $512$     & $\sim 443K$   \\
    \sys (Batch$^{B=4}$)            & $83.4$    & $128$     & $\sim 170K$   \\
    \sys (Batch$^{B=8}$)            & $85.4$    &  $64$     & $\sim 113K$   \\
    \sys (Batch$^{B=16}$)           & $89.3$    &  $32$     & $\sim  85K$   \\
    \sys (Batch$^{B=32}$)           & $90.6$    &  $16$     & $\sim  70K$   \\
    \bottomrule
  \end{tabularx}
  \vspace{-2mm}
\end{table}

\textbf{Varying Batch Sizes.} Herein, $B$ denotes the number of queries within a single batch. We increase the batch size until we reach the maximal context length of an LLM. As shown in Table~\ref{tab:topic-classification-batching}, batching significantly reduces both the number of required LLM queries and the total number of tokens submitted to the LLM. The larger the batch size, the more the saving. Moreover, larger batch size tends to lead to better accuracy. This is because a larger batch size provides more contextual information for the LLM to make better decisions. Therefore, our batching optimization is a win-win.

\begin{table}[htbp]
  \vspace{-2mm}
  \caption{Ablation Study on Batching Strategy.}
  \vspace{-4mm}
  \label{tab:topic-classification-batching-strategy}
  \begin{tabularx}{0.48\linewidth}{lcr}
    \toprule
    Method                             & OS        \\
    \midrule
    \sys (Batch$^{B=32}$)              & $90.6$    \\
    \sys (Batch$^{B=32,\texttt{DIV}}$) & $92.4$    \\
    \sys (Batch$^{B=32,\texttt{PRX}}$) & $91.6$    \\
    \sys (Batch$^{B=32,\texttt{CLS}}$) & $89.6$    \\
    \sys (Batch$^{B=32,\texttt{FAR}}$) & $90.8$    \\
    \bottomrule
  \end{tabularx}
  \vspace{-2mm}
\end{table}

\textbf{Different Batching Strategies.} In \sys, the default batching strategy adheres to the distribution of the dataset by randomly sampling queries to form a batch. Other batching strategies including $\texttt{DIV}$, $\texttt{PRX}$, $\texttt{CLS}$, $\texttt{FAR}$ are discussed in Sec.~\ref{sec:infra-batching}. As shown in Tab.~\ref{tab:topic-classification-batching-strategy}, different batching strategies lead to slightly different accuracies. Because the performance of the default random sampling strategy is already satisfactory, we recommend using the random sampling to save users' effort in selecting the best strategy. 
    
\end{sloppypar}

\textbf{Frozen Model with Different Distance Thresholds.} The OS dataset is a relatively easy dataset, which enables a large distance threshold for the CacheReuse. In The extreme case, when using a relatively large $d$ threshold, \sys reduces the ratio of LLM calls on OS by $95.5\%$, with an accuracy comparable to or even better than always querying GPT-4.

\begin{table}[htbp]
  \vspace{-2mm}
  \caption{Data Annotation: Varying Threshold $d$}
  \vspace{-4mm}
  \label{tab:topic-classification-ablation}
  \begin{tabularx}{0.85\linewidth}{lcrr}
    \toprule
    Method                              & OS        & $k$ ($\#$LLM)  & LLM R        \\
    \midrule
    \sys (LLM only)                     & $84.2$    & $512$          & $100.0\%$    \\
    \sys (CacheReuse$^{d=0.4}$)         & $86.1$    & $368$          &  $71.9\%$    \\
    \sys (CacheReuse$^{d=0.6}$)         & $88.3$    & $159$          &  $31.1\%$    \\
    \sys (CacheReuse$^{d=0.8}$)         & $87.9$    &  $64$          &  $12.5\%$    \\
    \sys (CacheReuse$^{d=1.0}$)         & $86.7$    &  $23$          &   $4.5\%$    \\
    \bottomrule
  \end{tabularx}
  \vspace{-2mm}
\end{table}

\section{Prompts}
\label{sec:appendix:prompts}

As the prompts in \sys are dynamically composed, we use $\texttt{<>}$ to represent placeholders for simplicity.

\subsection{Prompts Related to Direct LLM Querying}

\noindent\textbf{Task Profile.} The task profile contains the description of the task, inputs, outputs, and examples. This is the basic building block and will be repeatedly used in other prompts as a component. We will refer to it as $\texttt{<task\_profile>}$ in the future. We demonstrate the task profile with an entity resolution example.
\begin{lstlisting}
Given two products, determine whether they are the same product.
The input contains the following attributes:
- entity1: dict. It contains three attributes: `title`, `manufacturer`, `price`. `title` and `manufacturer` are strings, `price` is float.
- entity2: dict. Same as entity1.
You are expected to output:
- int. Output 0 if the two product are not identical, 1 of the two products are identical.
Examples:
Example #0:
Inputs:
- entity1: {'title': "mia 's math adventure : just in time", 'manufacturer': 'kutoka', 'price': 19.99}
- entity2: {'title': "kutoka interactive 61208 mia 's math adventure : just in time !", 'manufacturer': 'kutoka interactive', 'price': 24.99}
Output:
- 1
Example #1:
Inputs:
- entity1: {'title': 'adobe creative suite cs3 design standard upgrade [ mac ]', 'manufacturer': 'adobe', 'price': 399.0}
- entity2: {'title': '29300183dm adobe creative suite 3 design standard media tlp tlp nonprofit download -', 'manufacturer': nan, 'price': 20.97}
Output:
- 0
Example #2:
Inputs:
- entity1: {'title': 'instant immersion 33 languages', 'manufacturer': 'topics entertainment', 'price': 49.99}
- entity2: {'title': 'instant immersion 33 languages', 'manufacturer': nan, 'price': 47.36}
Output:
- 1
\end{lstlisting}

\noindent\textbf{LLM Query Prompt.} The LLM query prompt simply reuses the task profile and adds an instance. We demonstrate the task profile with a specific entity resolution instance.
\begin{lstlisting}
<task_profile>
Now consider the following instance:
- entity1: {'title': 'high school advantage 2008', 'manufacturer': 'encore', 'price': 39.99}
- entity2: {'title': 'elementary advantage 2008 encore', 'manufacturer': 'nan', 'price': 39.95}
Please respond with the answer only. Please do not output any other responses or any explanations.
\end{lstlisting}

\subsection{Prompts Related to Code Generation}

\noindent\textbf{Tools Profile.} This describes the tools which can either be used in code generation (e.g., using a third-party python package) or in the LLM module as interactive tools. Each tool is in the format of $\texttt{- <tools\_api>: <tools\_desc>}$. For example, in data discovery:
\begin{lstlisting}
- SUBMIT(table): Input a string. Submit the table.
- GET_SCHEMA(table_name): Input a string. Return the schema (the name of all the columns) of the given table.
- SEARCH_KEYWORDS(keywords): Input a list of keywords. Return a list containing at most 20 tables whose title or schema strictly contains at least one given keyword. The tables that contain more keywords will be ranked higher.
- SEARCH_VALUE(value): Input a value in any type. Return a list containing at most 20 tables that contains this value (the value is fuzzy matched as a string). The tables that contain more of this value will be ranked higher.
- BM25(query): Input a string. Return a list containing at most 20 tables that are related to the query, found by running bm25 algorithm over table title and schema.
- JOINT_SEARCH(keywords, value): Input a list of keywords and a value. Return a list containing at most 20 tables that contains both the value and at least one of the keywords.
\end{lstlisting}
When the tools profile is used in code generation a default profile allows the code to use any python package:
\begin{lstlisting}
- python packages: You can use any python packages you want. You do not need to install but only import them before using. You can not use supervised-learning method as there is no training data. Though, you can use frozen models if you want.
\end{lstlisting}
We will refer to the tools profile as $\texttt{<tools\_profile>}$ in the future.

\noindent\textbf{Advice Generation.}
\begin{lstlisting}
Please help me with the following Python programming task:
<task_profile>
<tools_profile>
Notice that the evaluation will severely punish incorrect outputs. Thus, when the function is uncertain, please return `None` to abstain instead of returning an incorrect guess.
The generated function should be robust, instead of only passing the provided examples. You are allowed use global variables to keep track of new incoming data.
Please provide a brief advice on how I should complete this programming task. Provide 2-3 concise sentences summarizing the key coding strategy.
\end{lstlisting}
We will refer to the advice generated as $\texttt{<advice>}$ in the future.

\noindent\textbf{Code Generation.}
\begin{lstlisting}
Please write a Python function `<task_api>` that completes the following goal:
<task_profile>
<tools_profile>
Hint: <advice>
Notice that the evaluation will severely punish incorrect outputs. Thus, when the function is uncertain, please return `None` to abstain instead of returning an incorrect guess.
The generated function should be robust, instead of only passing the provided examples. You are allowed use global variables to keep track of new incoming data.
Please respond with the Python implementation of `<task_api>` only. Please do not output any other responses or any explanations.
Your response should be in the following format (the markdown format string should be included):
```python
def <task_api>:
    '''Your Implementation Here.'''
```
\end{lstlisting}

\noindent\textbf{Logical Correctness Evaluation.}
\begin{lstlisting}
Consider the following task:
<task_profile>
Consider the following Python function `<task_api>` that is expected to complete the above task:
```python
<code>
```
Please determine whether the function `<task_api>` is correct.
Please output your judgement in the following format: first output "Thought:" and then thoughts on whether this function is correct or not, then output "Answer:" followed by a single answer "yes" or "no".
\end{lstlisting}

\noindent\textbf{Example (Test Case) Generation.} Each test case is viewed as an assertion statement in Python. We demonstrate the example generation using data imputation on the Buy dataset:
\begin{lstlisting}
Consider the following task:
<task_profile>
Consider the following Python function `data_imputation(name, desc)` that is expected to complete the above task:
```python
<code>
```
Please provide a 3 ~ 5 test cases that you think are effective and comprehensive for determining whether the program is correct.
You should focus on the logical correctness of the program.
Do not design corner cases such as small floating point errors, values beyond natural ranges, behaviors around boundary, etc.
The provided test cases should be multiple lines of compilable code, each line should be in the same format as below:
```python
assert data_imputation(name, desc) == ?
```
Example test cases:
```python
assert data_imputation("Linksys EtherFast EZXS55W Ethernet Switch", "5 x 10/100Base-TX LAN") == "Linksys"
assert data_imputation("PlayStation 2 Memory Card 8MB", "nan") == "Sony"
assert data_imputation("Directed Electronics SCH1 SiriusConnect Home Tuner", "SIRIUS SCH1 SIRIUSConnect Home Tuner") == "Sirius"
```
Now, please design some new test cases. Please respond with the test cases only.
\end{lstlisting}

Some generated test cases:
\begin{lstlisting}
assert data_imputation("Apple iPhone 13 Pro Max", "A15 Bionic chip with Neural Engine") == "Apple"
assert data_imputation("Samsung Galaxy Watch 4", "nan") == "Samsung"
assert data_imputation("Canon EOS R5 Mirrorless Camera", "45 Megapixel Full-Frame CMOS Sensor") == "Canon"
assert data_imputation("Bose QuietComfort 35 II Wireless Headphones", "nan") == "Bose"
assert data_imputation("LG 55-inch 4K OLED Smart TV", "Dolby Vision and Dolby Atmos") == "LG"
\end{lstlisting}

\noindent\textbf{Code Fixing Decision.}
\begin{lstlisting}
Consider the following task:
<task_profile>
Consider the following Python function `data_imputation(name, desc)` that is expected to complete the above task:
```python
<code>
```
The function seems to be incorrect:
<error_info>
Please determine whether: A. the function is actually incorrect and can be fixed. B. the function is actually incorrect and can not be easily fixed. C. the function is correct while the case evaluation is incorrect.
Please output your judgement in the following format: first output "Thought:" and then thoughts on whether this function is correct or not, then output "Answer:" followed by a single answer "A", "B" or "C".
\end{lstlisting}

\noindent\textbf{Code Fixing Advice Generation.}
\begin{lstlisting}
Consider the following task:
<task_profile>
Consider the following Python function `data_imputation(name, desc)` that is expected to complete the above task:
```python
<code>
```
The function seems to be incorrect:
<error_info>
Please provide a brief advice on how I should fix the function. Provide 2-3 concise sentences summarizing the key coding strategy.
The fix should be robust and general, instead of only passing the provided error cases.
\end{lstlisting}

\noindent\textbf{Code Fixing.}
\begin{lstlisting}
Consider the following task:
<task_profile>
Consider the following Python function `data_imputation(name, desc)` that is expected to complete the above task:
```python
<code>
```
The function seems to be incorrect:
<error_info>
Hint: <advice>
Please fix the code. The fix should be robust and general, instead of only passing the provided error cases.
Please respond with the fixed Python implementation of `<task_api>` only. Please do not output any other responses or any explanations.
Your response should be in the following format (the markdown format string should be included):
```python
def <task_api>:
    '''Your Implementation Here.'''
```
\end{lstlisting}

\subsection{Prompts Related to LLM Query}

\noindent\textbf{Batched Query Prompt.} The batched query prompt modifies the LLM query prompt by batching multiple instances. For simplicity, we demonstrate the batched query prompt with the entity resolution instance from the LLM query prompt and omit the other instances.
\begin{lstlisting}
<task_profile>
Now consider the following instances:
Instance #1:
- entity1: {'title': 'high school advantage 2008', 'manufacturer': 'encore', 'price': 39.99}
- entity2: {'title': 'elementary advantage 2008 encore', 'manufacturer': 'nan', 'price': 39.95}
Instance #2:
...
Please respond with the answer only, one line for each instance. Please do not output any other responses or any explanations.
Each response should start with "Output #<index>: ". For example:
Output #1: ...
Output #2: ...
\end{lstlisting}

\noindent\textbf{LLM with Tools Integration.} As the invocation process is iterative, we first show the initial prompt:
\begin{lstlisting}
<task_profile>
Please do not directly answer the problem. This should be an interactive process. You are allowed to take one of the following actions:
<tools_profile>
Interaction Examples:
<interaction_examples>
Now consider the following instance:
Your respond should strictly follow this format: first output `Thought:` followed by your thought process, then output `Action:` followed by one of the actions mentioned above.
\end{lstlisting}

Then we demonstrate the interaction process with a data discovery instance on Spider:
\begin{lstlisting}
...
Thought: I need to find tables related to the query, try with keywords 'employees', 'job', and 'past'.
Action: SEARCH_KEYWORDS(['employees', 'job', 'past'])
Observation: The top results: hr_1-employees, store_1-employees, phone_market-market, network_2-Person, hr_1-jobs, hr_1-job_history, department_store-Staff_Department_Assignments, department_management-department, customer_deliveries-Employees, cre_Doc_Tracking_DB-Employees, cre_Doc_Control_Systems-Employees, college_1-EMPLOYEE, yelp_user, yelp_tip, yelp_review, yelp_neighbourhood, yelp_checkin, yelp_category, yelp_business, wta_1-rankings

Thought: It seems `hr_1-employees`, `hr_1-jobs`, and `hr_1-job_history` are related to the query. I should confirm them by checking their schema.
Action: GET_SCHEMA('hr_1-employees')
Observation: The table has columns: `EMPLOYEE_ID`, `FIRST_NAME`, `LAST_NAME`, `EMAIL`, `PHONE_NUMBER`, `HIRE_DATE`, `JOB_ID`, `SALARY`, `COMMISSION_PCT`, `MANAGER_ID`, `DEPARTMENT_ID`

Thought: This table seems to be related to the query as it contains information about employees and their jobs. Now I need to check the schema of the other tables.
Action: GET_SCHEMA('hr_1-jobs')
Observation: The table has columns: `JOB_ID`, `JOB_TITLE`, `MIN_SALARY`, `MAX_SALARY`

Thought: This table seems to be related to the query as it contains information about jobs. Now I need to check the schema of the last table.
Action: GET_SCHEMA('hr_1-job_history')
Observation: The table has columns: `EMPLOYEE_ID`, `START_DATE`, `END_DATE`, `JOB_ID`, `DEPARTMENT_ID`

Thought: This table seems to be related to the query as it contains information about the job history of employees. It is time to submit the tables related to the query.
Action: SUBMIT(['hr_1-employees', 'hr_1-jobs', 'hr_1-job_history'])
\end{lstlisting}

\section{Generated Code}
\label{sec:appendix:code}

\subsection{Examples for Data Imputation}

\noindent\textbf{Example \#1}: a code snippet generated for the Buy dataset with LLM-generated advice.
\begin{lstlisting}[language=Python]
# Advice: Use a Named Entity Recognition (NER) model. NER is a subtask of information extraction that seeks to locate and classify named entities mentioned in unstructured text into pre-defined categories such as person names, organizations, locations, medical codes, time expressions, quantities, monetary values, percentages, etc. In this case, you can train or use a pre-trained NER model to recognize manufacturer names in the product name and description. This method requires a lot of computational resources and may not be feasible if you don't have a large amount of training data.
import spacy
nlp = spacy.load("en_core_web_sm")
def data_imputation(name, desc):
    doc = nlp(name)
    for ent in doc.ents:
        if ent.label_ == "ORG":
            return ent.text
    if desc != 'nan':
        doc = nlp(desc)
        for ent in doc.ents:
            if ent.label_ == "ORG":
                return ent.text
    return None
\end{lstlisting}

\noindent\textbf{Example \#2}: a code snippet generated for the Buy dataset with an alternate advice also generated by LLM.
\begin{lstlisting}[language=Python]
# Advice: Use a pre-defined dictionary of known manufacturers and their products. This method requires a lot of manual work to create and maintain the dictionary, but it can be very accurate. For example, you can define a dictionary where the keys are product names or keywords (like "PlayStation") and the values are the corresponding manufacturers (like "Sony"). Then, you can search for these keywords in the product name and description to deduce the manufacturer.
def data_imputation(name, desc):
    manufacturers = {
        "Linksys": ["EtherFast", "EZXS55W", "Ethernet Switch"],
        "Sony": ["PlayStation"],
        "Sirius": ["SiriusConnect", "Home Tuner"]
    }
    for manufacturer, keywords in manufacturers.items():
        for keyword in keywords:
            if keyword in name or (desc != 'nan' and keyword in desc):
                return manufacturer
    return None
\end{lstlisting}

\noindent\textbf{Example \#3}: a code snippet branched from failure cases. More manufacturers are included in the list.
\begin{lstlisting}[language=Python]
def data_imputation(name, desc):
    import re
    # List of common manufacturers
    manufacturers = ['Linksys', 'Sony', 'Sirius', 'Apple', 'Samsung', 'Microsoft', 'Dell', 'HP', 'Lenovo', 'Asus', 'Acer', 'LG', 'Panasonic', 'Philips', 'Canon', 'Nikon', 'Bose', 'JBL', 'Logitech', 'Netgear', 'Seagate', 'Western Digital', 'Kingston', 'SanDisk', 'Corsair', 'Intel', 'AMD', 'Nvidia', 'Qualcomm', 'Broadcom']
    # Check if manufacturer is in the name
    for manufacturer in manufacturers:
        if manufacturer.lower() in name.lower():
            return manufacturer
    # Check if manufacturer is in the description
    if desc != 'nan':
        for manufacturer in manufacturers:
            if manufacturer.lower() in desc.lower():
                return manufacturer
    # If manufacturer is not found, return None
    return None
\end{lstlisting}

\noindent\textbf{Example \#4}: a code snippet generated for the Restaurant dataset with LLM-generated advice and improved upon failure cases.
\begin{lstlisting}[language=Python]
# Use Geocoding APIs: You can use a geocoding API, such as Google's Geocoding API, to convert the restaurant's address into geographic coordinates, and then extract the city name from the returned data. This approach is more complex and requires internet access, but it is likely to be the most accurate, especially if the restaurant's name and phone number do not contain the city name.
import re
import phonenumbers
import usaddress
from geopy.geocoders import Nominatim
def data_imputation(name, addr, phone, type):
    # Extract city from restaurant name
    city_in_name = re.search(r'\((.*?)\)', name)
    if city_in_name:
        return city_in_name.group(1).lower()
    # Extract city from phone number
    try:
        parsed_phone = phonenumbers.parse(phone, 'US')
        if parsed_phone.country_code == 1 and parsed_phone.national_number:
            city_in_phone = phonenumbers.geocoder.description_for_number(parsed_phone, 'en')
            if city_in_phone:
                return city_in_phone.lower()
    except:
        pass
    # Extract city from address using Geocoding API
    geolocator = Nominatim(user_agent="geoapiExercises")
    location = geolocator.geocode(addr)
    if location:
        address_components = usaddress.parse(location.address)
        for component in address_components:
            if 'PlaceName' in component:
                return component[0].lower()
    return None
\end{lstlisting}

\subsection{Examples for Information Extraction}

\noindent\textbf{Example \#5}: a code snippet generated for the Wiki NBA dataset {\it{without}} LLM-generated advice.
\begin{lstlisting}[language=Python]
from bs4 import BeautifulSoup
import re
def information_extraction_college(html):
    soup = BeautifulSoup(html, 'html.parser')
    for table in soup.find_all('table', {'class': 'infobox'}):
        for row in table.find_all('tr'):
            th = row.find('th')
            if th and 'college' in th.text.lower():
                td = row.find('td')
                if td:
                    college = re.sub(r'\[[^]]*\]', '', td.text) # Remove any reference tags
                    return college.strip()
    return None
\end{lstlisting}

\noindent\textbf{Example \#6}: a code snippet generated for the Wiki NBA dataset with LLM-generated advice. The main difference is that the class containing the relevant information is not ``infobox'' but ``infobox vcard'', which is addressed in the advice.
\begin{lstlisting}[language=Python]
# You can use BeautifulSoup, a Python library for parsing HTML and XML documents. It creates parse trees that are helpful to extract the data easily. You can search for the player's college information by finding the table with class 'infobox vcard' which usually contains the player's information on Wikipedia. Then, iterate over the rows of the table until you find a row with the header 'College'. The next cell in that row should contain the college information. If no such row is found, return None.
from bs4 import BeautifulSoup
def information_extraction_college(html):
    soup = BeautifulSoup(html, 'html.parser')
    table = soup.find('table', {'class': 'infobox vcard'})
    if table:
        for row in table.find_all('tr'):
            if row.find('th') and 'College' in row.find('th').text:
                return row.find('td').text.split('(')[0].strip()
    return None
\end{lstlisting}
\end{document}